%% file: draft.tex
\documentclass[conference, twocolumn]{IEEEtran}
\IEEEoverridecommandlockouts
\usepackage{cite}
\usepackage{amsmath, amssymb, mathtools, bm, gensymb}
\usepackage[shortcuts]{glossaries}
\usepackage{algorithmic}
\usepackage{graphicx}
\usepackage{textcomp}
\usepackage{xcolor}
\usepackage{pgfplots}
\usepackage{amsthm}
\usepackage{float}

\usetikzlibrary{arrows}
\usetikzlibrary{shapes}
\pgfplotsset{compat=1.14}

\def\BibTeX{{\rm B\kern-.05em{\sc i\kern-.025em b}\kern-.08em
    T\kern-.1667em\lower.7ex\hbox{E}\kern-.125emX}}

\renewcommand{\baselinestretch}{0.94}

\newcommand{\Mbs}{\ensuremath{M_{\text{BS}}}}
\newcommand{\Mh}{\ensuremath{M_{\text{H}}}}
\newcommand{\Mv}{\ensuremath{M_{\text{V}}}}
\newcommand{\Etx}{\ensuremath{E_{\text{Tx}}}}
\newcommand{\hhlos}{\ensuremath{\bm{h}_{\text{H},u}^\text{LOS}[n]}}
\newcommand{\hvlos}{\ensuremath{\bm{h}_{\text{V},u}^\text{LOS}[n]}}
\DeclareMathOperator*{\argmax}{arg\,max}
\DeclareMathOperator*{\vspan}{span}

\newtheorem{theorem}{Theorem}
\newtheorem{corollary}{Corollary}
\newtheorem{lemma}{Lemma}

\newtheorem*{nremark}{Remark}

\input{./acronyms.tex}
\input{./mathmacros.tex}
\input{./macros.tex}

\begin{document}

\title{Low-Complexity Massive MIMO Tensor Precoding
}

\author{
	\IEEEauthorblockN{\emph{Lucas N. Ribeiro$^*$, Stefan Schwarz$^{\dagger}$, Andr\'e L. F. de Almeida$^\ddagger$, Martin Haardt$^*$}}
	\IEEEauthorblockA{$^*$Communications Research Laboratory, Technische Universit\"at (TU) Ilmenau, Ilmenau, Germany\\
	$^{\dagger}$Christian Doppler Laboratory for Dependable Wireless Connectivity for the
	Society in Motion,
	TU Wien, Vienna, Austria\\
	$^\ddagger$Wireless Telecommunications Research Laboratory, Federal University of Cear\'a, Fortaleza, Brazil
}
}

\maketitle

\begin{abstract}
	We present a novel and low-complexity massive multiple-input multiple-output (MIMO) precoding strategy based on novel findings concerning the subspace separability of Rician fading channels. Considering a uniform rectangular array at the base station,  we show that the subspaces spanned by the channel vectors can be factorized as a tensor product between two lower dimensional subspaces. Based on this result, we formulate tensor maximum ratio transmit and zero-forcing precoders. We show that the proposed tensor precoders exhibit lower computational complexity and require less instantaneous channel state information than their linear counterparts. Finally, we  present computer simulations that demonstrate the applicability of the proposed tensor precoders in practical communication scenarios.
\end{abstract}


\section{Introduction}

Massive \ac{mimo} is one of the main enabling technologies for 5G networks. It consists of employing many active antennas at the \ac{bs} to serve multiple users on the same time-frequency resource~\cite{marzetta_fundamentals_2016}. It can provide high data throughput by spatial signal processing at the \ac{bs} (precoding) to combat multi-user interference and to provide large beamforming gains. \Ac{mrt} and \ac{zf} precoding are known to perform well if accurate instantaneous \ac{csi} is available at the \ac{bs}. However, in practice, \ac{csi} estimates are often noisy, since accurate estimation of the high-dimensional massive \ac{mimo} channels can be quite expensive in terms of power and time-frequency resources. Moreover, the \ac{zf} precoder is known to be computationally expensive due to the large number of computations it requires to invert the high-dimensional channel Gram matrix~\cite{van_der_perre_efficient_2018}.

Several solutions have been proposed to simplify the \ac{csi} requirements and to reduce the computational complexity of massive \ac{mimo} precoders. An efficient solution consists of designing precoders based on partial \ac{csi} \cite{ying_kronecker_2014,qiu_downlink_2018,schwarz_robust_2018}. This kind of \ac{csi} is typically less expensive to estimate than full instantaneous \ac{csi}. To reduce the computational complexity of the precoder design, many strategies are available in the literature. For example, series expansion techniques~\cite{zhang_performance_2018}, precoder interpolation~\cite{kashyap_frequency-domain_2016}, decentralized filtering~\cite{li_decentralized_2017}, and multi-layer filtering~\cite{alkhateeb_multi-layer_2017,ribeiro_double-sided_2019}. Different approaches that exploit the algebraic properties of the \ac{mimo} channel to reduce both \ac{csi} requirements and computational complexity have been investigated in the literature~\cite{ribeiro_identification_2015,ribeiro_low-complexity_2017,ribeiro_separable_2019,ribeiro_low-rank_2019,wang_two-dimensional_2017,zhu_hybrid_2017}. In some conditions, the channel may be well-approximated by the tensor product between lower dimensional components. This allows us to develop low-complexity tensor filters.

Tensor filtering has been applied to system identification~\cite{ribeiro_identification_2015} and equalization problems~\cite{ribeiro_low-complexity_2017,ribeiro_separable_2019,ribeiro_low-rank_2019}. A common aspect among these works is the separable system model, i.e., the vector or matrix that models the system can be exactly factorized in terms of tensor products. In~\cite{ribeiro_identification_2015,ribeiro_low-complexity_2017,ribeiro_separable_2019,ribeiro_low-rank_2019}, we have developed low-complexity tensor filtering schemes that exploit this property to reduce the number of calculations involved in the filter design. However, strict separability is rarely encountered in practice due to the non-separable nature of many devices and physical phenomena, limiting the applicability of the previously proposed filtering methods.

In this paper, we adopt a different approach to our previous works. Instead of assuming simplified separable models, we consider a practical channel model and demonstrate that, under some conditions, the subspaces spanned by the channel vectors can be factorized into a tensor product between lower dimensional subspaces. More specifically, we consider a \ac{bs} equipped with a \ac{upa} and we assume Rician fading channels. We show that the subspace spanned by the channel vectors can be decomposed into the tensor product between the subspaces spanned by the \ac{bs} horizontal and vertical linear sub-arrays. Based on this result, we formulate the \ac{tmrt} and \ac{tzf} precoders. These tensor precoders are based on low-dimensional instantaneous \ac{csi}. Therefore,  they are less expensive to estimate than the full high-dimensional \ac{csi} required by the classical \ac{mrt} and \ac{zf} precoders. Moreover, they require much less computational resources than their classical counterparts. For example, we show that the runtime of \ac{tzf} is twice as fast as that of \ac{zf} while the achievable sum-rate difference is minimal. 

The proposed precoders are related to the techniques discussed in \cite{wang_two-dimensional_2017,zhu_hybrid_2017}. The work of \cite{wang_two-dimensional_2017} proposes a precoder scheme that exploits the geometry of URAs by means of the Kronecker (tensor)  product to reduce the design complexity. Likewise, \cite{zhu_hybrid_2017} leverages the tensor product in the design of analog beamformers to null undesired signals. We emphasize that the present paper extends the contributions of \cite{wang_two-dimensional_2017,zhu_hybrid_2017} by providing novel  theoretical results about the channel subspace separability. Moreover, we show when these results can be applied to reduce the \ac{csi} requirements and the computational complexity of the proposed precoding methods considering a practical dynamical scenario.

\emph{Notation:} Vectors and matrices are written as lowercase and uppercase boldface letters, respectively. The transpose and the conjugate transpose (Hermitian) of $\bm{X}$ are represented by $\bm{X}^\tran$ and $\bm{X}^\hermit$, respectively. The $N$-dimensional identity matrix is represented by $\bm{I}_{N}$ and the $(M\times N)$-dimensional null matrix by $\bm{0}_{M\times N}$. The symbol $\delta(\cdot)$ denotes the Kronecker's delta function. The $\Diag(\cdot)$ operator transforms an input vector into a diagonal matrix, $\vspan(\cdot)$ refers to the subspace spanned by the argument vectors, $O(\cdot)$ stands for the Big-O complexity notation, and $\otimes$ denotes the tensor product (also known as Kronecker product). The notation $[\bm{v}]_{\mathcal{I}}$ represents the vector obtained by selecting the elements of $\bm{v}$ that corresponds to the index set $\mathcal{I}$.

\section{System Model}

We consider a single-cell massive \ac{mimo} system with a \ac{bs} serving $U$ single-antenna user equipment (UEs). The system operates on perfectly-synchronized \ac{tdd} and the uplink-downlink channel reciprocity assumption holds. The BS is equipped with a \ac{upa} of size $\Mbs = \Mh \cdot \Mv$, as illustrated in Figure~\ref{fig:model}. Considering the downlink operation, the \ac{bs} employs precoding filters  $\bm{f}_u[n] \in \mathbb{C}^{\Mbs \times 1}$ to serve a data stream $s_u[n]$ to each UE $u$ at each \ac{tti} $n$. Let $\bm{h}_u[n] \in \mathbb{C}^{\Mbs \times 1}$ denote the downlink channel vector. Then, the received signal by the UE $u$ at \ac{tti} $n$ can be expressed as 
\begin{equation}
	y_u[n] = \bm{h}_u^\hermit[n]\bm{f}_u[n] s_u[n] + \sum_{j\neq u}^U \bm{h}_u^\hermit \bm{f}_{j}[n] s_j[n] + b_u[n],
\end{equation}
where $b_u[n]$ denotes a zero mean complex-valued \ac{awgn} component. We assume that $\mathbb{E}[s_u[m]s_j^*[n]] = \delta(u-j) \cdot \delta(m-n)$ and $\mathbb{E}[b_u[m]b_j^*[n]] = \sigma_b^2 \cdot \delta(u-j) \cdot \delta(m-n)$. The average \ac{bs} transmit power constraint can be expressed as $\sum_{u=1}^U E_{\text{Tx},u} \leq \Etx$, with $E_{\text{Tx},u} = \| \bm{f}_u[n] \|_2^2$ denoting the power allocated to UE $u$, and $\Etx \geq 0$ the total transmit power. The precoding filters are optimized based on imperfect \ac{csi}, as we will explain in more details in Section~\ref{sec:csi}. We define the downlink \ac{snr} as $\gamma_\text{DL} = \Etx/\sigma_b^2$.

\subsection{Channel Model}

We assume Rician flat fading channel model such that
\begin{equation}
	\bm{h}_u[n] = \sqrt{\frac{K}{K+1}} \bm{h}_u^\text{LOS}[n] + \sqrt{\frac{1}{K+1}} \bm{h}_u^{\text{NLOS}}[n] \in \mathbb{C}^{\Mbs \times 1},
\end{equation}
where $K \geq 0$ denotes the Rician $K$-factor, $\bm{h}_u^\text{LOS}[n]$ the \ac{los} component, and $\bm{h}_u^{\text{NLOS}}[n]$ the \ac{nlos} component. Note that $K$ controls the influence  of the \ac{los} term over the \ac{nlos} one. 

\begin{figure}[t]
	\centering
	\includegraphics[width=0.3\textwidth]{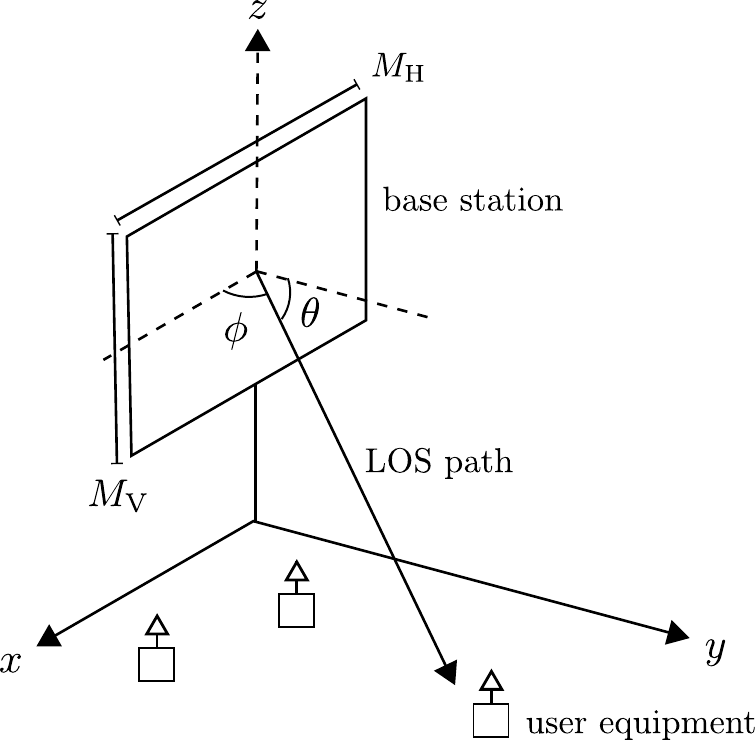}
	\caption{Illustration of the \ac{bs} and its geometry. Angles $\phi$ and $\theta$ denote  azimuth and elevation, respectively.}
	\label{fig:model}
\end{figure}

The \ac{los} component is determined by a Doppler phase shift $\psi_u[n]$ and an array steering vector $\bm{a}_{u}[n] \in \mathbb{C}^{\Mbs \times 1}$. Both components depend on the UE location and on its velocity relative to the \ac{bs}. Let us first define some geometrical notation to describe the Doppler phase shift. Let $\bm{p}_{\text{BS}}[n] \in \mathbb{R}^3$ and $\bm{p}_{\text{UE},u}[n] \in \mathbb{R}^3$ denote the $3$-dimensional position vectors of the \ac{bs} and UE $u$, respectively. The \ac{bs}-UE $u$ distance vector is defined as $\bm{d}_u[n] = \bm{p}_{\text{UE},u}[n] - \bm{p}_{\text{BS}}[n]$, and is normalized as
\begin{equation}
	\tilde{\bm{d}}_u[n] = \frac{\bm{d}_u[n]}{\| \bm{d}_u[n] \|_2} = \left[\tilde{d}_u^x[n], \tilde{d}_u^y[n], \tilde{d}_u^z[n]\right]^\tran.
\end{equation}
As illustrated in Figure~\ref{fig:model}, the respective elevation and azimuth angles of UE $u$ at \ac{tti} $n$ are given by
\begin{gather} \label{eq:mean_angles}
	\overline{\theta}_u[n] = \arcsin \tilde{d}_u^z[n], \quad \overline{\phi}_u[n] = \arctan\left( \frac{\tilde{d}_u^y[n]}{\tilde{d}_u^x[n]} \right).
\end{gather}
Assuming a certain random angle spread, the geometrical elevation and azimuth angles in \eqref{eq:mean_angles} are modeled as 
\begin{gather}
	\theta_u[n] = \overline{\theta}_u[n] + X_\theta[n] \label{eq:lostheta}\\
	\phi_u[n] = \overline{\phi}_u[n] + X_\phi[n], \label{eq:losphi}
\end{gather}
where $X_\theta[n]$ and $X_\phi[n]$ denote real-valued independent and identically distributed Gaussian random variables with zero mean and variance $\sigma_\theta^2$ and $\sigma_\phi^2$, respectively. Define the \ac{los} wave vector as~\cite{van_trees_optimum_2002}
\begin{align}
	&\bm{k}_u[n] =\\ 
	&\tfrac{2\pi}{\lambda} [\cos\theta_u[n]\cos\phi_u[n], \cos\theta_u[n]\sin \phi_u[n], \sin \theta_u[n]]^\tran, \nonumber
\end{align}
where $\lambda$ denotes the carrier wavelength. Furthermore, let $\bm{v}_u[n] \in \mathbb{R}^3$ denote the speed vector of UE $u$ relative to the fixed \ac{bs}. The Doppler phase shift is finally defined as $\psi_u[n] = \bm{k}_u^\tran[n] \bm{v}_u[n]$. The steering vector definition depends on how the antenna array elements are arranged in space. Considering a \ac{upa} placed in the $x$-$z$ plane as illustrated in Figure~\ref{fig:model}, the $m$-th element of the steering vector is given by~\cite{van_trees_optimum_2002}
	\begin{gather}
		[\bm{a}_u[n]]_m = \sqrt{g_{u,m}[n]} \cdot e^{-\jmath(\delta_{m_{\text{H}}}^{(u)}[n] +  \xi_{m_{\text{V}}}^{(u)}[n])} \label{eq:steering_el_full}\\
		\delta_{m_{\text{H}}}^{(u)}[n] = \frac{2\pi}{\lambda} d_\text{H}(m_\text{H} - 1) \cos \theta_u[n] \cos \phi_u[n] \\
		\xi_{m_{\text{V}}}^{(u)}[n] = \frac{2\pi}{\lambda} d_\text{V}(m_\text{V} - 1) \sin \theta_u[n]\\
		m = m_\text{V} + (m_\text{H} - 1) \cdot M_\text{V}\\
		m_\text{V} \in \{ 1,\ldots,M_\text{V}\}, \quad m_\text{H} \in \{ 1,\ldots, M_\text{H}\},
	\end{gather}
with $g_{u,m}[n]$ representing the $m$-th antenna element gain, and $d_{\text{H}}$ and $d_\text{V}$ the horizontal and vertical inter-antenna spacing, respectively. Note that the antenna gain $g_{u,m}[n]$ is a function of $\theta_u[n]$ and $\phi_u[n]$. From \eqref{eq:steering_el_full}, it follows that
\begin{equation} \label{eq:steering_vec_sep}
	\bm{a}_u[n] = \bm{G}_u[n] (\bm{a}_{\text{H},u}[n] \otimes \bm{a}_{\text{V},u}[n]),
\end{equation}
where $ \bm{G}_u[n] = \Diag(\sqrt{g_{u,1}[n]}, \ldots, \sqrt{g_{u,\Mbs}[n]})$ stands for the $\Mbs$-dimensional diagonal antenna gains matrix. The vectors $\bm{a}_{\text{H},u}[n] \in \mathbb{C}^{\Mh \times 1}$ and $\bm{a}_{\text{V},u}[n] \in \mathbb{C}^{\Mv \times 1}$ represent the horizontal and vertical sub-array steering vectors, respectively. Their elements are defined as
\begin{equation}
	[\bm{a}_{\text{H},u}[n]]_{m_\text{H}} = e^{-\jmath \delta_{m_{\text{H}}}^{(u)}[n]}, \quad [\bm{a}_{\text{V},u}[n]]_{m_\text{V}} = e^{-\jmath \xi_{m_{\text{V}}}^{(u)}[n]}.
\end{equation}
for $m_\text{H} \in \{ 1,\ldots,\Mh\}$, and  $m_\text{V} \in \{ 1,\ldots, \Mv\}$. Finally, the \ac{los} component can be expressed as
\begin{subequations}
	\begin{align}
		\bm{h}_{u}^\text{LOS}[n] &= e^{\jmath \psi_u[n]} \cdot \bm{a}_u[n]\\
								 &= e^{\jmath \psi_u[n]} \cdot \bm{G}_u[n] (\bm{a}_{\text{H},u}[n] \otimes \bm{a}_{\text{V},u}[n]).
	\end{align}	
\end{subequations}

The \ac{nlos} component consists of diffuse background scattering components modeled as Rayleigh fading. To model the fading time evolution, the fading is modeled as a first-order Gauss-Markov process \cite{truong_effects_2013}. Hence, the \ac{nlos} component is given by
\begin{equation}
	\bm{h}_{u}^{\text{NLOS}}[n+1] = \rho_u[n] \bm{h}_{u}^{\text{NLOS}}[n] + \sqrt{1-\rho_u^2[n]}\bm{h}_{\mathcal{N}},
\end{equation}
where $\bm{h}_{\mathcal{N}}$ is a \ac{zmcsg} random vector with spatial covariance matrix $\bm{R}_{\mathcal{N}}$, and $\rho_u[n]$ denotes the temporal correlation parameter. Considering the Clarke-Jakes autocorrelation model, it follows that $\rho_u[n] = J_0(2\pi f_{\text{D},u}[n] T_s)$, with $J_0(\cdot)$ representing the zeroth-order Bessel function, $f_{\text{D},u}[n] = \|\bm{v}_u[n]\|_2/ \lambda$, the maximum Doppler shift, and $T_s$ the \ac{tti} length.

\subsection{Sub-Array Representation}

The algebraic structure of \eqref{eq:steering_vec_sep} allows us to obtain the individual contributions of $\bm{a}_{\text{H},u}[n]$ and $\bm{a}_{\text{V},u}[n]$ by carefully selecting the elements of $\bm{h}_u[n]$. We define the respective horizontal and vertical sub-array index sets as
	\begin{gather}
		\mathcal{I}_{\text{H}} = \{ 1 + (m_\text{H}-1)\Mv\, | \, m_h = 1,\ldots, M_\text{H} \}\\
		\mathcal{I}_{\text{V}} = \{1, \ldots, \Mv \}.
	\end{gather}
Also, we define the $\Mh$-dimensional horizontal sub-array channel vector as
{\small
\begin{gather}
	\bm{h}_{\text{H},u}[n] = [ \bm{h}_u[n] ]_{\mathcal{I}_{\text{H}}} = \sqrt{\frac{K}{K+1}} \hhlos + \sqrt{\frac{1}{K+1}} \bm{h}_{\text{H},u}^{\text{NLOS}}[n] \\
	\bm{h}_{\text{H},u}^\text{LOS}[n] = e^{\jmath \psi_u[n]} \cdot \bm{G}_{\text{H}}[n] \bm{a}_{\text{H},u}[n], \quad \bm{h}_{\text{H},u}^{\text{NLOS}}[n] = \left[ \bm{h}_{u}^{\text{NLOS}}[n] \right]_{\mathcal{I}_\text{H}} \nonumber\\
	\bm{G}_{\text{H}}[n] = \Diag( \sqrt{g_{u,m}[n]} ), \forall m \in \mathcal{I}_{\text{H}}\nonumber
\end{gather}}
and the $\Mv$-dimensional vertical sub-array channel vector as
{\small
\begin{gather}
	\bm{h}_{\text{V},u}[n] = [ \bm{h}_u[n] ]_{\mathcal{I}_{\text{V}}} = \sqrt{\frac{K}{K+1}} \hvlos + \sqrt{\frac{1}{K+1}} \bm{h}_{\text{V},u}^{\text{NLOS}}[n] \\
	\bm{h}_{\text{V},u}^\text{LOS}[n] = e^{\jmath \psi_u[n]} \cdot \bm{G}_{\text{V}}[n] \bm{a}_{\text{V},u}[n], \quad \bm{h}_{\text{V},u}^{\text{NLOS}}[n] = \left[ \bm{h}_{u}^{\text{NLOS}}[n] \right]_{\mathcal{I}_\text{V}} \nonumber\\
	\bm{G}_{\text{V}}[n] = \Diag( \sqrt{g_{u,m}[n]} ), \forall m \in \mathcal{I}_{\text{V}}\nonumber
\end{gather}}
Note that the sub-array channel vectors are obtained from $\bm{h}_u[n]$ by simply selecting the corresponding vector elements.

\subsection{CSI Acquisition} \label{sec:csi}

In \ac{tdd} systems with calibrated \ac{rf} front-ends, the downlink channels are reciprocal to the uplink channels. Therefore, the \ac{bs} may obtain channel estimates from pilot sequences transmitted in uplink training slots. Let $\bm{p}_u = [p_u[0],\ldots,p_u[L-1]]^\tran$ denote the length-$L$  pilot sequence of UE $u$. We assume that the pilot sequences follow an orthogonal design, i.e., $\bm{p}_i^\hermit \bm{p}_j = L \cdot \delta(i-j)$. This orthogonality property can be found in many sequences, for example, \ac{dft} and Zadoff-Chu sequences. During the uplink training \ac{tti} $n$, the UEs simultaneously transmit their pilot sequences to the BS with power $E_\text{P}$. Thus, the received signal at the \ac{bs} can be written as
\begin{equation} \label{eq:ul_signals}
	\bm{X}[n] = \sqrt{E_\text{P}}\sum_{u=1}^U \bm{h}_u[n] \bm{p}_u^\hermit[n] + \bm{B}[n] \in \mathbb{C}^{\Mbs \times L},
\end{equation}
with $\bm{B}[n] \in \mathbb{C}^{\Mbs \times L}$ denoting the uplink complex-valued \ac{awgn} term. The elements of the noise matrix are modeled as \ac{zmcsg} random variables with variance $\sigma_b^2$. From \eqref{eq:ul_signals}, the \ac{ls} estimate of the UE $u$ channel vector is given by
\begin{equation} \label{eq:ls_model}
	\hat{\bm{h}}_u[n] = \frac{1}{L \sqrt{E_\text{P}}}\bm{X}[n] \bm{p}_u[n] = \bm{h}_u[n] + \frac{1}{L \sqrt{E_\text{P}}} \bm{B}[n] \bm{p}_u[n].
\end{equation}
We define the uplink \ac{snr} as $\gamma_{\text{UL}} = E_\text{P}/\sigma_b^2$.

\section{Precoding Methods}

This section begins with a brief review of the classical \ac{mrt} and \ac{zf} precoders. Then, these classical precoding schemes are reformulated considering the tensor approach in Section~\ref{sec:tenfilt}. Unlike previous works \cite{ribeiro_low-rank_2019} that rely on the explicit channel separability, the proposed \ac{tmrt} and \ac{tzf} precoders are based on the tensor factorization of the intended and interfering UEs' subspaces. Our results on the factorization of these subspaces are the main theoretical contribution of this paper and they are discussed in Theorems~\ref{theo:1} and \ref{theo:2}. Finally, the \ac{csi} requirements and the computational complexity of the proposed precoders are discussed in Section~\ref{sec:comp}.

\subsection{Linear Precoders}

\subsubsection{Maximum Ratio Transmition (MRT)}
The \ac{mrt} precoder $\bm{f}_{\text{MRT}, u}[n]$ is designed to maximize the received signal power at the intended user \cite{bjornson_optimal_2014}. From the Cauchy-Schwarz inequality, the \ac{mrt} precoder is given by
\begin{equation}
	\bm{f}_{\text{MRT}, u}[n] = \argmax_{\bm{f}, \|\bm{f}\|_2^2=E_{\text{Tx},u}} \left|\bm{h}_u[n]^\hermit \bm{f}\right|^2 = \frac{\sqrt{E_{\text{Tx},u}}}{\| \bm{h}_u[n] \|_2} \bm{h}_u[n],
\end{equation}
Note that the \ac{mrt} precoder does not attempt to cancel the multi-user interference.

\subsubsection{Zero-Forcing (ZF)}

The \ac{zf} precoder $\bm{f}_{\text{ZF}, u}[n]$ is designed to satisfy the zero multi-user interference condition:
\begin{equation} \label{eq:zf}
	\tilde{\bm{H}}_u[n] \bm{f}_{\text{ZF}, u}[n] = \bm{0}_{(U-1)\times 1},
\end{equation}
where 
\begin{align}
	\tilde{\bm{H}}_u[n] = [ \bm{h}_1[n], \ldots, \bm{h}_{u-1}[n], \bm{h}_{u+1}[n], \ldots, \bm{h}_U[n] ]^\hermit \label{eq:int}
\end{align}
denotes the  $(U-1)\times \Mbs$-dimensional multi-user interference channel matrix relative to UE $u$. This condition can be satisfied by projecting the \ac{mrt} precoder onto the null-space of the matrix $\tilde{\bm{H}}_u[n]$ if $\Mbs \geq U$~\cite{bdhaardt}. To this end, consider the following eigenvalue decomposition:
\begin{equation} \label{eq:gramint}
	\tilde{\bm{H}}_u^\hermit[n]\tilde{\bm{H}}_u[n] = \bm{V}_u[n] \bm{\Lambda}_u[n] \bm{V}_u^\hermit[n],
\end{equation}
with $\bm{V}_u[n] \in \mathbb{C}^{\Mbs \times \Mbs}$ denoting the eigenvector matrix, and $\bm{\Lambda}_u[n] \in \mathbb{C}^{\Mbs \times \Mbs}$ the diagonal eigenvalue matrix. The null-space projector can be expressed as
\begin{equation} \label{eq:zf_projector}
	\bm{P}_{u}[n] = \bm{I}_{\Mbs} - \tilde{\bm{V}}_u[n] \tilde{\bm{V}}_u^\hermit[n] \in \mathbb{C}^{\Mbs \times \Mbs},
\end{equation}
where $\tilde{\bm{V}}_u[n] \in \mathbb{C}^{\Mbs \times (U-1)}$ is formed by the $U-1$ dominant eigenvectors of \eqref{eq:gramint}. The \ac{zf} precoder is then given by:
\begin{gather}
	\tilde{\bm{f}}_u[n] = \bm{P}_{u}[n] \bm{h}_u[n]\\
	\bm{f}_{\text{ZF},u}[n] = \frac{\sqrt{E_{\text{Tx},u}}}{\| \tilde{\bm{f}}_u[n] \|_2} \tilde{\bm{f}}_u[n].
\end{gather}

\subsection{Tensor Precoders} \label{sec:tenfilt}

\subsubsection{Tensor Maximum Ratio Transmission (TMRT)}

The \ac{tmrt} precoder is built upon the separability of the subspace spanned by $\bm{h}_u[n]$ presented in Theorem~\ref{theo:1}. To support the demonstration of this result, we present the following two lemmas.
	
\begin{lemma} \label{lemma:vec}
	For sufficiently large Rician-$K$ factors, $\bm{h}_u[n]$, $\bm{h}_{\text{H},u}[n]$, and $\bm{h}_{\text{V},u}[n]$ may be well-approximated as
	{\small\begin{align}
		\bm{h}_u[n] &\approx \bm{h}_{u}^{\text{LOS}}[n] = e^{\jmath \psi_u[n]} \cdot \bm{G}_u[n] (\bm{a}_{\text{H},u}[n] \otimes \bm{a}_{\text{V},u}[n]) \label{eq:approx1}\\
		\bm{h}_{\text{H},u}[n] &\approx \bm{h}_{\text{H}, u}^{\text{LOS}}[n] = e^{\jmath \psi_u[n]} \cdot \bm{G}_{\text{H},u}[n] \bm{a}_{\text{H},u}[n] \label{eq:approx2}\\
		\bm{h}_{\text{V},u}[n] &\approx \bm{h}_{\text{V}, u}^{\text{LOS}}[n] = e^{\jmath \psi_u[n]} \cdot \bm{G}_{\text{V},u}[n] \bm{a}_{\text{V},u}[n] \label{eq:approx3}
		\end{align}}
\end{lemma}
\begin{proof}
	The \ac{nlos} terms become insignificant relative to the \ac{los} terms for sufficiently large Rician-$K$ factors. Therefore, the considered approximation holds.
\end{proof}

\begin{lemma} \label{lemma:vec2}
	The respective basis vectors of the subspaces spanned by the approximations in \eqref{eq:approx1}--\eqref{eq:approx3} are given by $\bm{a}_{\text{H},u}[n] \otimes \bm{a}_{\text{V},u}[n]$, $\bm{a}_{\text{H},u}[n]$, and $\bm{a}_{\text{V},u}[n]$ if the diagonal antenna gain matrices $\bm{G}_u[n]$, $\bm{G}_{\text{H},u}[n]$, and $\bm{G}_{\text{V},u}[n]$ are non-singular.
\end{lemma}
\begin{proof}
	This result follows from the fact that the Doppler shift and the antenna gain matrices do not change the direction of the corresponding steering vectors.
\end{proof}

\begin{theorem} \label{theo:1}
	For sufficiently large Rician-$K$ factors, the subspace spanned by $\bm{h}_u[n]$ consists of the tensor product between the subspaces generated by $\bm{h}_{\text{H},u}[n]$ and $\bm{h}_{\text{V},u}[n]$.
\end{theorem}
\begin{proof}
	Lemmas \ref{lemma:vec} and \ref{lemma:vec2} demonstrate that the basis vector of the subspace spanned by $\bm{h}_u[n]$ is given by the tensor product between $\bm{a}_{\text{H},u}[n]$ and $\bm{a}_{\text{V},u}[n]$. Furthermore, these steering vectors form bases for the subspaces generated by $\bm{h}_{\text{H},u}[n]$ and $\bm{h}_{\text{V},u}[n]$, respectively.
\end{proof}

\begin{nremark}
	The channel vector $\bm{h}_u[n]$ cannot be factorized into a tensor product between $\bm{h}_{\text{H},u}[n]$ and $\bm{h}_{\text{V}, u}[n]$ in general because
	\begin{enumerate}
		\item The Rayleigh \ac{nlos} term $\bm{h}_u^{\text{NLOS}}[n]$ does not have any specific separable structure;
		\item The antenna gains matrix $\bm{G}_u[n]$ cannot be decomposed into elevation and azimuth factors.
	\end{enumerate}
	However, Theorem \ref{theo:1} shows that the $1$-dimensional subspace spanned by  $\bm{h}_u[n]$ can be factorized if the conditions provided in Lemmas~\ref{lemma:vec} and \ref{lemma:vec2} are satisfied, as illustrated in Figure~\ref{fig:subspace}.
\end{nremark}

\begin{figure}[t]
	\centering
	\includegraphics[width=0.4\textwidth]{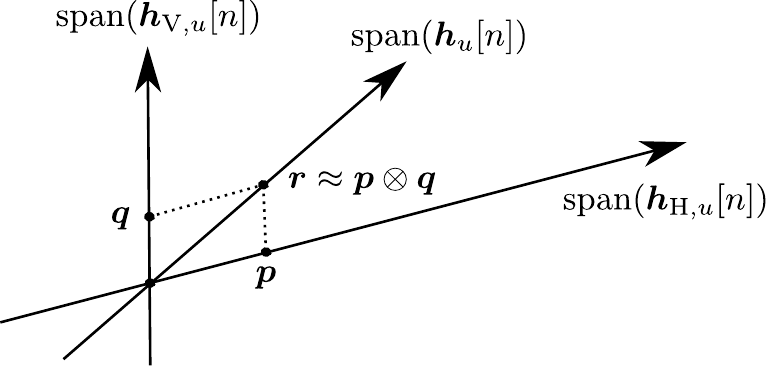}
	\caption{Illustration of Theorem~\ref{theo:1}. The vector $\bm{r} \in \vspan(\bm{h}_u[n])$ can be well approximated by the tensor product between $\bm{p} \in \vspan(\bm{h}_{\text{H},u}[n])$ and $\bm{q} \in \vspan(\bm{h}_{\text{V},u}[n])$.}
	\label{fig:subspace}
\end{figure}

Based on Theorem~\ref{theo:1}, the \ac{tmrt} precoder consists of designing sub-array precoders $\tilde{\bm{f}}_{\text{H}, u}[n]$ and $\tilde{\bm{f}}_{\text{V}, u}[n]$ that maximize the received power at the intended UE and then  combining them through the tensor product as

{\small
\begin{gather}
	\bm{f}_{\text{TMRT}, u}[n] = \tilde{\bm{f}}_{\text{H}, u}[n] \otimes \tilde{\bm{f}}_{\text{V}, u}[n]\\
	\tilde{\bm{f}}_{\text{H}, u}[n] = \argmax_{\bm{f}_\text{H}, \|\bm{f}_\text{H}\|_2^2 = E_{\text{Tx},u}^{1/2}} |\bm{h}_{\text{H},u}[n]\bm{f}_\text{H}|^2 = \tfrac{\sqrt[4]{E_{\text{Tx},u}}}{\|\bm{h}_{\text{H},u}[n]\|_2 } \bm{h}_{\text{H},u}[n]\\
	\tilde{\bm{f}}_{\text{V}, u}[n] = \argmax_{\bm{f}_\text{V}, \|\bm{f}_\text{V}\|_2^2 = E_{\text{Tx},u}^{1/2}} |\bm{h}_{\text{H},u}[n]\bm{f}_\text{V}|^2 = \tfrac{\sqrt[4]{E_{\text{Tx},u}}}{\|\bm{h}_{\text{V},u}[n]\|_2 } \bm{h}_{\text{V},u}[n].
\end{gather}}

\begin{figure*}[!t]
	{\footnotesize\begin{gather}
			\tilde{\bm{H}}_u^\hermit[n]\tilde{\bm{H}}_u[n] = \sum_{j\neq u}^U \left( \frac{K}{K+1} \bm{h}_{j}^{\text{LOS}}[n] \bm{h}_{j}^{\text{LOS}}[n]^\hermit + \frac{\sqrt{K}}{K+1} \bm{h}_{j}^{\text{LOS}}[n] \bm{h}_{j}^{\text{NLOS}}[n]^\hermit + \frac{\sqrt{K}}{K+1} \bm{h}_{j}^{\text{NLOS}}[n] \bm{h}_{j}^{\text{LOS}}[n]^\hermit + \frac{1}{K+1} \bm{h}_{j}^{\text{NLOS}}[n] \bm{h}_{j}^{\text{NLOS}}[n]^\hermit\right) \label{eq:int_gram}\\
			\tilde{\bm{H}}_{\text{H},u}^\hermit[n]\tilde{\bm{H}}_{\text{H},u}[n] = \sum_{j\neq u}^U \left( \frac{K}{K+1} \bm{h}_{\text{H},j}^{\text{LOS}}[n] \bm{h}_{\text{H},j}^{\text{LOS}}[n]^\hermit + \frac{\sqrt{K}}{K+1} \bm{h}_{\text{H},j}^{\text{LOS}}[n] \bm{h}_{\text{H},j}^{\text{NLOS}}[n]^\hermit + \frac{\sqrt{K}}{K+1} \bm{h}_{\text{H},j}^{\text{NLOS}}[n] \bm{h}_{\text{H},j}^{\text{LOS}}[n]^\hermit + \frac{1}{K+1} \bm{h}_{\text{H},j}^{\text{NLOS}}[n] \bm{h}_{\text{H},j}^{\text{NLOS}}[n]^\hermit\right) \label{eq:h_gram} \\
			\tilde{\bm{H}}_{\text{V},u}^\hermit[n]\tilde{\bm{H}}_{\text{V},u}[n] = \sum_{j\neq u}^U \left( \frac{K}{K+1} \bm{h}_{\text{V}, j}^{\text{LOS}}[n] \bm{h}_{\text{V}, j}^{\text{LOS}}[n]^\hermit + \frac{\sqrt{K}}{K+1} \bm{h}_{\text{V}, j}^{\text{LOS}}[n] \bm{h}_{\text{V}, j}^{\text{NLOS}}[n]^\hermit + \frac{\sqrt{K}}{K+1} \bm{h}_{\text{V}, j}^{\text{NLOS}}[n] \bm{h}_{\text{V}, j}^{\text{LOS}}[n]^\hermit + \frac{1}{K+1} \bm{h}_{\text{V}, j}^{\text{NLOS}}[n] \bm{h}_{\text{V}, j}^{\text{NLOS}}[n]^\hermit\right) \label{eq:v_gram}
		\end{gather}}
		\hrulefill
\end{figure*}

\subsubsection{Tensor Zero-Forcing (TZF)}

The \ac{tzf} precoder is formulated based on the results of Theorem~\ref{theo:2} and of Corollary~\ref{cor:1}. To support these results, we present Lemmas~\ref{lemma:kint}--\ref{lemma:col_int}. For future convenience, let us define the horizontal and vertical multi-user interference channel matrices
	\begin{gather}
		\tilde{\bm{H}}_{\text{H},u}[n] = [ \bm{h}_{\text{H},1}[n], \ldots, \bm{h}_{\text{H},u-1}[n], \bm{h}_{\text{H},u+1}[n], \ldots, \bm{h}_{\text{H},U}[n] ]^\hermit\\
		\tilde{\bm{H}}_{\text{V},u}[n] = [ \bm{h}_{\text{V},1}[n], \ldots, \bm{h}_{\text{V},u-1}[n], \bm{h}_{\text{V},u+1}[n], \ldots, \bm{h}_{\text{V},U}[n] ]^\hermit	
	\end{gather}
with dimensions $(U-1) \times \Mh$ and $(U-1)\times \Mv$, respectively, and the Gram matrices in \eqref{eq:int_gram}--\eqref{eq:v_gram}, shown on the top of the next page.
\begin{lemma} \label{lemma:kint}
	For sufficiently large Rician-$K$ factors, the Gram matrices \eqref{eq:int_gram}--\eqref{eq:v_gram} can be well approximated as
	\begin{gather}
	\tilde{\bm{H}}_u^\hermit[n]\tilde{\bm{H}}_u[n] \approx \sum_{j\neq u}^U  \bm{G}_j[n] \bm{R}_j[n] \bm{G}_j[n] \label{eq:gramapprox}\\
	\tilde{\bm{H}}_{\text{H},u}^\hermit[n]\tilde{\bm{H}}_{\text{H},u}[n] \approx \sum_{j\neq u}^U  \bm{G}_{\text{H},j}[n] \bm{R}_{\text{H},j}[n] \bm{G}_{\text{H},j}[n] \label{eq:gramapprox_h}\\
	\tilde{\bm{H}}_{\text{V},u}^\hermit[n]\tilde{\bm{H}}_{\text{V},u}[n] \approx \sum_{j\neq u}^U  \bm{G}_{\text{V},j}[n] \bm{R}_{\text{V},j}[n] \bm{G}_{\text{V},j}[n],\label{eq:gramapprox_v}
	\end{gather}
	with 
	\begin{gather}
	\bm{R}_j[n] = \bm{R}_{\text{H},j}[n] \otimes \bm{R}_{\text{V},j}[n] \in \mathbb{C}^{\Mbs \times \Mbs} \label{eq:rjkron}\\
	\bm{R}_{\text{H},j}[n] = \bm{a}_{\text{H},j}[n] \bm{a}_{\text{H},j}[n]^\hermit \in \mathbb{C}^{\Mh \times \Mh} \\
	\bm{R}_{\text{V},j}[n] = \bm{a}_{\text{V},j}[n] \bm{a}_{\text{V},j}[n]^\hermit \in \mathbb{C}^{\Mv \times \Mv}
	\end{gather}
\end{lemma}
\begin{proof}
	The approximations in \eqref{eq:gramapprox}--\eqref{eq:gramapprox_v} are obtained by noticing that the \ac{nlos} and \ac{nlos}-\ac{los} cross terms in \eqref{eq:int_gram}--\eqref{eq:v_gram} are negligible compared to the \ac{los} component for sufficiently large Rician-$K$ factors. In \eqref{eq:gramapprox}, $\bm{R}_j[n]$ is obtained by applying the mixed-product property of the tensor product~\cite{sidiropoulos_tensor_2017} as:
	\begin{subequations}
	\begin{align}
		\bm{R}_j[n] &= \left( \bm{a}_{\text{H},j}[n] \otimes \bm{a}_{\text{V},j}[n] \right) \left( \bm{a}_{\text{H},j}[n] \otimes \bm{a}_{\text{V},j}[n] \right)^\hermit\\
					&= \left( \bm{a}_{\text{H},j}[n] \bm{a}_{\text{H},j}[n]^\hermit \right) \otimes \left( \bm{a}_{\text{V},j}[n] \bm{a}_{\text{V},j}[n]^\hermit \right)\\
					&= \bm{R}_{\text{H},j}[n] \otimes \bm{R}_{\text{V},j}[n] 
	\end{align}
	\end{subequations}
\end{proof}
\begin{lemma} \label{lemma:diag_int}
	The approximations \eqref{eq:gramapprox}--\eqref{eq:gramapprox_v} and the matrices 
	\begin{gather}
	\bm{\Theta}_u[n] = \sum_{j\neq u}^U \bm{R}_j[n] \label{eq:thetakron}\\
	\bm{\Theta}_{\text{H},u}[n] = \sum_{j\neq u}^U \bm{R}_{\text{H},j}[n]\label{eq:thetakronh}\\
	\bm{\Theta}_{\text{V},u}[n] = \sum_{j\neq u}^U \bm{R}_{\text{V},j}[n]\label{eq:thetakronv}
	\end{gather}
	span the same column-spaces, respectively, if the diagonal antenna gains matrices $\bm{G}_j[n]$, $\bm{G}_{\text{H},j}[n]$, and $\bm{G}_{\text{V},j}[n]$ are non-singular.
\end{lemma}
\begin{proof}
	The non-singular scaling performed by the diagonal antenna gains matrices in \eqref{eq:gramapprox}--\eqref{eq:gramapprox_v} do not change the direction of the eigenvectors of $\bm{R}_{j}[n]$, $\bm{R}_{\text{H},j}[n]$, and $\bm{R}_{\text{V},j}[n]$.
\end{proof}
\begin{lemma} \label{lemma:col_int}
	The column-space of $\bm{\Theta}_u[n]$ can be factorized into the tensor product between the column-spaces of $\bm{\Theta}_{\text{H},u}[n]$ and $\bm{\Theta}_{\text{V},u}[n]$ if the elevation angles that locate the interfering UEs are approximately equal.
\end{lemma}
\begin{proof}
	By inserting \eqref{eq:rjkron} into \eqref{eq:thetakron}, we have
	\begin{equation} \label{eq:thetaeita}
	\bm{\Theta}_u[n] = \sum_{j\neq u}^U  \bm{R}_{\text{H},j}[n] \otimes \bm{R}_{\text{V},j}[n].
	\end{equation}
	If the interfering users' elevation angles $\theta_j[n]$ are alike, then the matrices $\bm{R}_{\text{V},j}[n]$ are almost equivalent for all $j \neq u$. This assumption allows us to approximately factorize \eqref{eq:thetaeita} as
	\begin{equation}
	\bm{\Theta}_u[n] \approx \left( \sum_{j\neq u}^U  \bm{R}_{\text{H},j}[n] \right) \otimes \bm{\Theta}_{\text{V},u}[n] = \bm{\Theta}_{\text{H},u}[n] \otimes \bm{\Theta}_{\text{V},u}[n].		
	\end{equation}
\end{proof}
\begin{theorem} \label{theo:2}
	If the elevation angles that locate the interfering UEs are sufficiently close and if the Rician-$K$ factor is sufficiently large, then the column-space of $\tilde{\bm{H}}_u^\hermit[n]\tilde{\bm{H}}_u[n]$ is approximately formed by the tensor product of the column-spaces of $\tilde{\bm{H}}_{\text{H},u}^\hermit[n]\tilde{\bm{H}}_{\text{H},u}[n]$ and $\tilde{\bm{H}}_{\text{V},u}^\hermit[n]\tilde{\bm{H}}_{\text{V},u}[n]$.
\end{theorem}
\begin{proof}
	Lemmas~\ref{lemma:kint} and \ref{lemma:diag_int} demonstrate that the column-spaces of \eqref{eq:gramapprox}--\eqref{eq:gramapprox_v} and \eqref{eq:thetakron}--\eqref{eq:thetakronv} are approximately the same for sufficiently large Rician-$K$ factors, respectively. Lemma \ref{lemma:col_int} shows that the subspaces of \eqref{eq:gramapprox}--\eqref{eq:gramapprox_v} are connected through a tensor product. Therefore, the column-space of $\tilde{\bm{H}}_u^\hermit[n]\tilde{\bm{H}}_u[n]$ is spanned by the tensor product between $\tilde{\bm{H}}_{\text{H},u}^\hermit[n]\tilde{\bm{H}}_{\text{H},u}[n] $ and $\tilde{\bm{H}}_{\text{V},u}^\hermit[n]\tilde{\bm{H}}_{\text{V},u}[n]$.
\end{proof}
\begin{corollary} \label{cor:1}
	Let the matrices $\bm{K}$, $\bm{K}_\text{H}$, and $\bm{K}_\text{V}$ denote projectors to the column-spaces of $\tilde{\bm{H}}_u^\hermit[n]\tilde{\bm{H}}_u[n]$, $\tilde{\bm{H}}_{\text{H},u}^\hermit[n]\tilde{\bm{H}}_{\text{H},u}[n] $, and  $\tilde{\bm{H}}_{\text{V},u}^\hermit[n]\tilde{\bm{H}}_{\text{V},u}[n]$, respectively. If the conditions of  Theorem~\ref{theo:2} hold, then $\bm{K} \approx \bm{K}_\text{H} \otimes \bm{K}_\text{V}$.
\end{corollary}
As in the classical \ac{zf} precoder, we build a projector to the null-space of the multi-user interference channel matrix $\tilde{\bm{H}}_u[n]$ using the results of Theorem 2 and Corollary 1. First, consider the eigenvalue decompositions of \eqref{eq:h_gram} and \eqref{eq:v_gram}
\begin{gather}
	\tilde{\bm{H}}_{\text{H},u}^\hermit[n]\tilde{\bm{H}}_{\text{H},u}[n] = \bm{V}_{\text{H},u}[n] \bm{\Lambda}_{\text{H},u}[n] \bm{V}_{\text{H},u}^\hermit[n] \label{eq:mat1}\\
	\tilde{\bm{H}}_{\text{V},u}^\hermit[n]\tilde{\bm{H}}_{\text{V},u}[n] = \bm{V}_{\text{V},u}[n] \bm{\Lambda}_{\text{V},u}[n] \bm{V}_{\text{V},u}^\hermit[n], \label{eq:mat2}
\end{gather}
and the corresponding column-space projectors
\begin{gather}
	\bm{K}_{\text{H},u}[n] = \tilde{\bm{V}}_{\text{H},u}[n] \tilde{\bm{V}}^\hermit_{\text{H},u}[n] \in \mathbb{C}^{\Mh \times \Mh}\\
	\bm{K}_{\text{V},u}[n] = \tilde{\bm{V}}_{\text{V},u}[n] \tilde{\bm{V}}^\hermit_{\text{V},u}[n] \in \mathbb{C}^{\Mv \times \Mv},
\end{gather}
where $\tilde{\bm{V}}_{\text{H},u}[n] \in \mathbb{C}^{\Mh \times (U-1)}$ and $\tilde{\bm{V}}_{\text{V},u}[n] \in \mathbb{C}^{\Mv \times (U-1)}$ contain the $U-1$ dominant eigenvectors of $\bm{V}_{\text{H},u}[n]$ and $\bm{V}_{\text{V},u}[n]$, respectively. From Corollary 1, the column-space projector of $\tilde{\bm{H}}_u^\hermit[n]\tilde{\bm{H}}_u[n]$ may be approximated as $\bm{K} \approx \bm{K}_\text{H} \otimes \bm{K}_\text{V}$. Then, the null-space projector \eqref{eq:zf_projector} can be approximated as
\begin{equation}
	\bm{P}_u[n] \approx \bm{P}_{\text{TZF},u}[n] = \bm{I}_{\Mbs} - \bm{K}_{\text{H},u}[n] \otimes \bm{K}_{\text{V},u}[n].
\end{equation}
Motivated by Theorem~\ref{theo:1}, we project the tensor product $\bm{h}_{\text{H},u}[n] \otimes \bm{h}_{\text{V,u}}[n]$ onto the multi-user interference channel matrix's null space. Hence, the \ac{tzf} precoder is then given by
\begin{gather}
\overline{\bm{f}}_u[n] = \bm{P}_{\text{TZF},u}[n] (\bm{h}_{\text{H},u}[n] \otimes \bm{h}_{\text{V,u}}[n])\\
\bm{f}_{\text{TZF},u}[n] = \frac{\sqrt{E_{\text{Tx},u}}}{\| \overline{\bm{f}}_u[n] \|_2} \overline{\bm{f}}_u[n].
\end{gather}
Note that the null spaces of \eqref{eq:mat1} and \eqref{eq:mat2} exist if and only if $\Mh$ and $\Mv$ are larger than $U-1$. Therefore, the feasibility condition for the \ac{tzf} precoder can be formulated as $\min(\Mh, \Mv) > U-1$.

\begin{figure}
	\centering
	\input{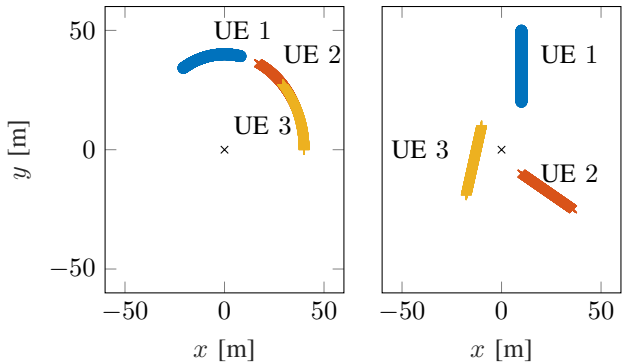}
	\caption{Circular (left) and linear (right) tracks. The BS array is located at $(0,0,25)$~m and the height of each UE is $1.5$ m.}
	\label{fig:tracks}
\end{figure}

\subsection{CSI Requirements}
\label{sec:csitreq}

The \ac{csi} typically required by linear precoders is the intended UE channel vector $\bm{h}_u[n]$ and possibly the interfering ones $\{\bm{h}_{j}[n]\,|\, j =1,\ldots,U,\,j\neq u\}$. This information can be obtained as detailed in Section~\ref{sec:csi}, for example. By contrast the proposed tensor precoders are based only on the $\Mh$- and $\Mv$-dimensional sub-array channel vectors $\bm{h}_{\text{H},u}[n]$ and $\bm{h}_{\text{V},u}[n]$, respectively. The total number of parameters to be estimated by the tensor approach is proportional to $\Mh + \Mv$, while that number is proportional to  $\Mh\cdot\Mv$ in the linear approach. Therefore, the tensor approach is less expensive than the linear one, as the \ac{csi} acquisition cost is usually proportional to the number of channel coefficients.

\subsection{Complexity Analysis} \label{sec:comp}

The \ac{mrt} precoder performs $2\Mbs + 1$ operations, and the \ac{tmrt} precoder carries out additional $\Mbs$ multiplications due to the tensor product, $3\Mbs + 1$ in total. Therefore, the linear and tensor \ac{mrt} precoders are comparable in terms of number of computations. The computational complexity of the \ac{zf}-based precoders is dominated by the calculation of the projection matrices. The number of multiplications required by this calculation is cubic with the channel vector length. Therefore, the \ac{zf} precoder performs $O(\Mbs^3) = O((\Mh\cdot\Mv)^3)$ multiplications, whereas the \ac{tzf} only $O(\Mh^3) + O(\Mv^3)$.

\section{Simulations}

In this section, we present simulation results to evaluate the channel subspace separability results and the proposed tensor precoders. We assume a single \ac{bs} with a half-wavelength \ac{upa} of $\Mbs = 16 \times 16$ isotropic  antennas serving $U=3$ single-antenna UEs. The \ac{tdd} frames are divided into uplink and downlink slots with \ac{tti} of $1$ ms, and the carrier frequency is $6$ GHz. The uplink and downlink \ac{snr} are $20$ dB and~$10$~dB, respectively, and the \ac{los} angle spreads in \eqref{eq:lostheta}, \eqref{eq:losphi} are $\sigma_\theta = \sigma_\phi =1^\circ$. We consider the circular and linear tracks depicted in Figure~\ref{fig:tracks} for the UE trajectories. The modulus of the UE speed vectors is constant $\| \bm{v}_u \|_2= 30$ m/s for $u=1,2,3$. The elevation and azimuth angles corresponding to the trajectory of UE $1$ during $1000$ TTIs ($=1$ second) are shown in Figure~\ref{fig:angles}. In the circular track, we observe that the elevation angle remains constant, whereas the azimuth angle grows. This track is therefore useful for investigating the evolution of the horizontal and vertical subspaces. The linear track describes a more realistic scenario, where both azimuth and elevation angles grow.

\begin{figure}
	\centering
	\input{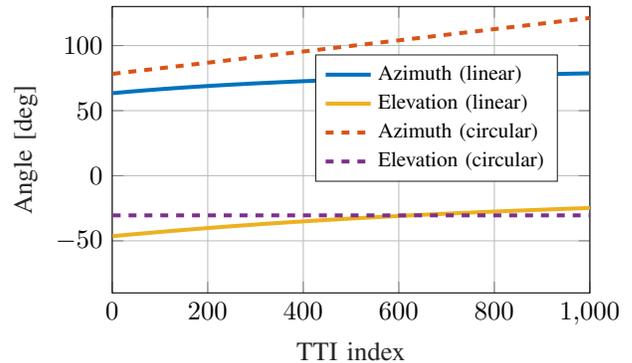}
	\caption{Angles for linear and circular tracks of UE $1$.}
	\label{fig:angles}
\end{figure}

In the first experiment scenario, we investigate the evolution of subspaces spanned by the channel vectors $\bm{h}_u[n]$, $\bm{h}_{\text{H},u}[n]$, and $\bm{h}_{\text{V},u}[n]$. To quantify the subspace evolution, we consider the chordal distance between the dominant eigenvectors of the channel correlation matrices. More specifically, the chordal distance between the eigenvectors $\bm{u}[m]$ and $\bm{u}[n]$ of same length is defined as $d_\text{C}^2(\bm{u}[m], \bm{u}[n]) = 1 - \left| \bm{u}^\hermit[m] \bm{u}[n] \right|^2$~\cite{schwarz_performance_2016}. In Figure~\ref{fig:dchord:ue}, the chordal distance is calculated relative to the eigenvector at \ac{tti} $n=0$ for the subspaces of the full, horizontal, and vertical channel vectors. The eigenvectors are estimated by averaging $256$ channel realizations for each \ac{tti}. The chordal distances relative to $\bm{h}_1[n]$, $\bm{h}_{\text{H},1}[n]$, and $\bm{h}_{\text{V},1}[n]$ are plotted in Figure~\ref{fig:dchord:ue} for Rician-$K$ factors of $0$~dB and $20$~dB considering the circular track. This result suggests that the vertical component, which corresponds to the elevation angle, remains constant, while the full and horizontal components evolve in time. This behavior is observed for both Rician-$K$ factors. We extend this experiment to evaluate the column-space evolution of the multi-user interference channel matrices \eqref{eq:int_gram}--\eqref{eq:v_gram}. These matrices are calculated relative to UE $1$ assuming the circular track scenario and $K=20$ dB for all UEs, as depicted in Figure~\ref{fig:tracks}. This result indicates that the column-space of the multi-user interfering channel matrix does not evolve. This is because the interfering UEs are characterized by the same elevation angle. Therefore, the approximation in Lemma~\ref{lemma:kint} holds and the conditions given in Theorem~\ref{theo:2} are satisfied.

\begin{figure}
	\centering
	\input{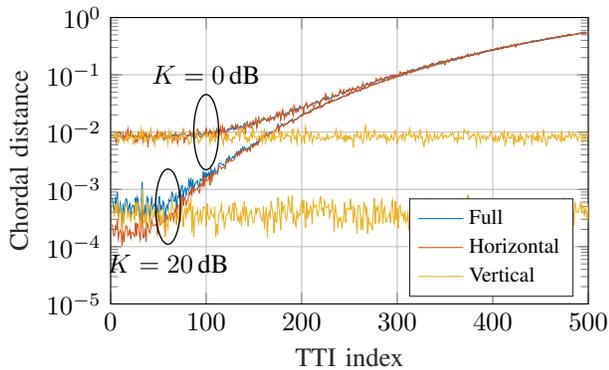}
	\caption{Chordal distance evolution of channel subspace with circular tracks.}
	\label{fig:dchord:ue}
\end{figure}
~
\begin{figure}
	\centering
	\input{./fig2-20db-int.tex}
	\caption{Chordal distance of multi-user interference channel subspace with circular tracks,  $K=20$ dB.}
	\label{fig:dchord:int}
\end{figure}
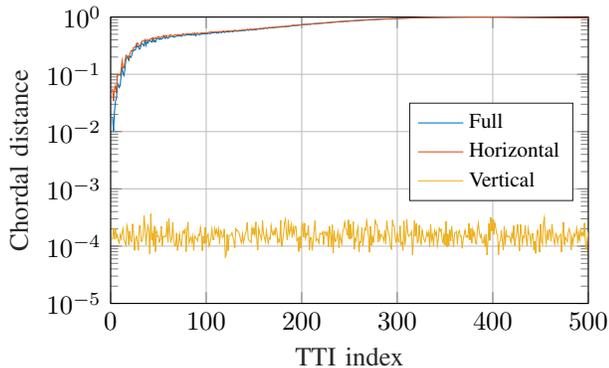

In the second experiment scenario, we evaluate the proposed tensor precoders for both circular and linear tracks and $K=20$ dB. As observed in Figures~\ref{fig:dchord:ue} and \ref{fig:dchord:int} for the circular track, the elevation component evolves slowly relative to the azimuth component. In this case, it is reasonable to think that it is not necessary to update the vertical precoders as often as the horizontal precoders. To test this hypothesis, we carried out experiments where the precoders are designed at uplink slots (allocated every $250$ TTIs) and employed in the successive downlink slots. At each uplink \ac{tti}, the \ac{bs} acquires the \ac{csi} of each UE as described in Section \ref{sec:csi}, designs the benchmark linear precoder (MRT or ZF) and the corresponding tensor precoder. We consider two implementations of the tensor precoders. The standard implementation consists of updating both horizontal (H) and vertical (V) tensor precoder filters at each uplink slot. The second implementation, by contrast, calculates the vertical component at the first uplink slot and, afterward, updates only the horizontal component.

In Figures~\ref{fig:mrt:circ} and \ref{fig:zf:circ}, the MRT- and ZF-based precoders are evaluated for circular tracks, respectively. As observed in Figures~\ref{fig:dchord:ue} and \ref{fig:dchord:int}, the elevation subspace does not change in circular tracks. Consequently, implementations of the \ac{tmrt} and \ac{tzf} precoders perform the same. Figure~\ref{fig:mrt:circ} shows that the linear and tensor approaches exhibit the same performance, empirically confirming the validity of Theorem~\ref{theo:1}. Regarding the ZF-based precoders in Figure~\ref{fig:zf:circ}, we observe that the tensor precoders exhibit a small loss compared to the linear~precoders.

\begin{figure}
	\centering
	\input{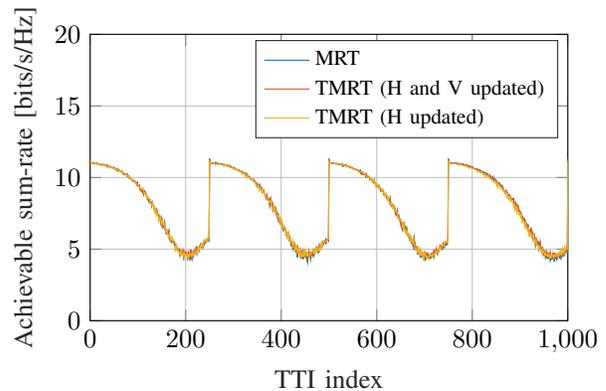}
	\caption{MRT-based precoders performance with circular tracks.}
	\label{fig:mrt:circ}
\end{figure}
~
\begin{figure}	
	\centering
	\input{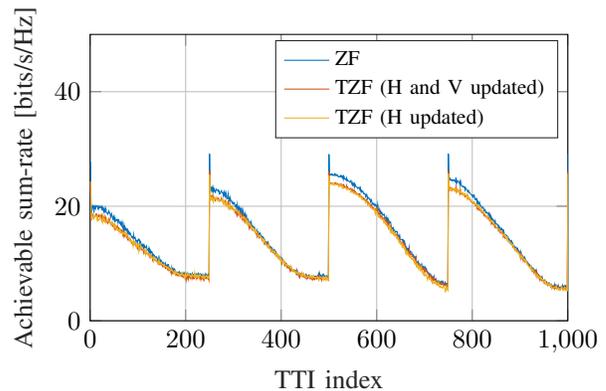}
	\caption{ZF-based precoders performance with circular tracks.}
	\label{fig:zf:circ}
\end{figure}	

Figures~\ref{fig:mrt:lin} and \ref{fig:zf:lin} show the performance of the MRT- and ZF-based precoders for linear tracks, respectively. In this scenario, both elevation and azimuth components evolve in time, and therefore, the corresponding precoders need to be updated at each uplink slot. These figures indicate that the tensor precoders face an important performance loss when only the horizontal precoders are updated. However, the loss relative to their linear counterparts is insignificant when both horizontal and vertical components are updated. These results demonstrate that the proposed tensor precoders may be applied to any kind of UE track to approximate their linear counterpart when both horizontal and vertical components are updated.

In the previous experiments, we have analyzed the precoders in terms of the achievable sum-rate. The ZF-based precoders are evaluated in terms of the empirical \ac{cdf} of the precoder runtime in Figure~\ref{fig:runtime}. This figure reveals that the tensor approach is roughly twice  as fast as its linear counterpart. Therefore, \ac{tzf} strikes an excellent rate-complexity trade-off.

	\begin{figure}[t]
		\centering
		\input{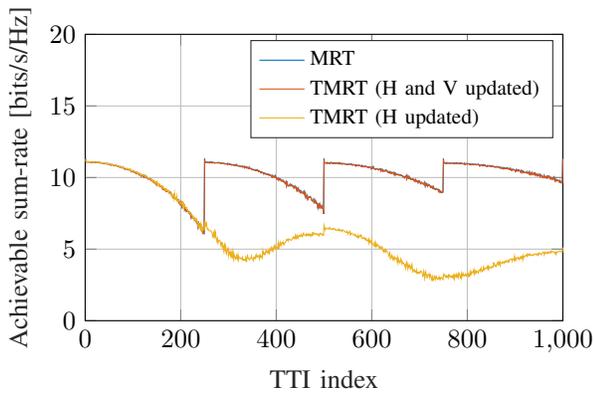}
		\caption{MRT-based precoders performance with linear tracks.}
		\label{fig:mrt:lin}
	\end{figure}
	\begin{figure}[t]
		\centering
		\input{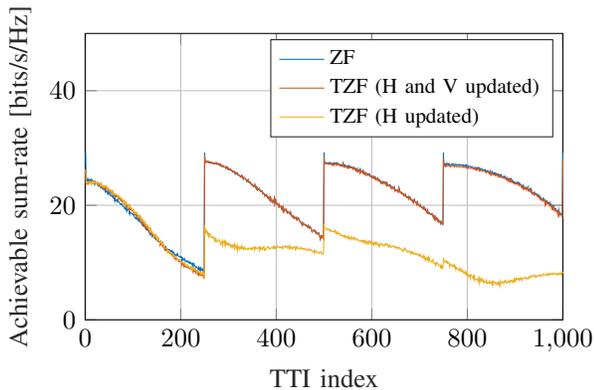}
		\caption{ZF-based precoders performance with linear tracks.}
		\label{fig:zf:lin}
	\end{figure}

\section{Conclusion}

In this paper, we present novel results concerning the subspace separability of Rician channels, and, based on these results, we propose efficient tensor precoders. We assess the validity of the subspace separability and the performance of the proposed precoders through computer simulations. The \ac{tmrt} and \ac{tzf} precoders can closely approximate their linear counterparts while exhibiting a much lower computational complexity. Specifically, the \ac{tzf} precoder is twice as fast as the \ac{zf} precoder, while exhibiting a negligible rate loss in a very dynamic communication scenario. 

\begin{figure}[t]
	\centering
	\input{./fig5-zf-runtime.tex}
	\caption{Empirical \ac{cdf} of runtime.}
	\label{fig:runtime}
\end{figure}
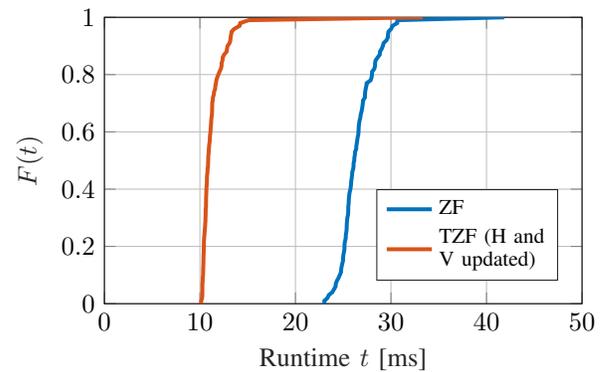

We considered a single specular \ac{los} component, which is adequate for specific scenarios, i.e., millimeter wave systems with a strong \ac{los} component. We plan to extend the proposed tensor precoding framework to manage communication scenarios with multiple specular components and multi-antenna receivers.

\bibliographystyle{IEEEtran}
\bibliography{./draft}

\end{document}

%% file: acronyms.tex
\newacronym{mimo}{MIMO}{multiple-input multiple-output}
\newacronym{als}{ALS}{alternating least squares}
\newacronym{bs}{BS}{base station}
\newacronym{3gpp}{3GPP}{third generation partnership project}
\newacronym{zf}{ZF}{zero-forcing}
\newacronym{mse}{MSE}{mean square error}
\newacronym{mmse}{MMSE}{minimum mean square error}
\newacronym{lms}{LMS}{least mean square}
\newacronym{upa}{URA}{uniform rectangular array}
\newacronym[plural=ULAs,firstplural=uniform linear arrays (ULAs)]{ula}{ULA}{uniform linear array}
\newacronym{zmcsg}{ZMCSG}{zero mean circularly symmetric Gaussian}
\newacronym{awgn}{AWGN}{additive white Gaussian noise}
\newacronym{snr}{SNR}{signal to noise ratio}
\newacronym{sir}{SIR}{signal to interference ratio}
\newacronym{sinr}{SINR}{signal to interference and noise ratio}
\newacronym{ber}{BER}{bit error ratio}
\newacronym{4g}{4G}{fourth-generation}
\newacronym{5g}{5G}{fifth-generation}
\newacronym{6g}{6G}{sixth-generation}
\newacronym{iid}{i.i.d.}{independent and identically distributed}
\newacronym{flops}{FLOPS}{floating point operations}
\newacronym{qpsk}{QPSK}{quadrature phase shift-keying}
\newacronym{lcmv}{LCMV}{linear constrained minimum variance}
\newacronym{af}{AF}{array factor}
\newacronym{music}{MUSIC}{multiple signal classification}
\newacronym{mc}{MC}{Monte Carlo}
\newacronym[plural=UEs, firstplural=user equipment (UE)]{ue}{UE}{user equipment}
\newacronym{cp}{CP}{canonical polyadic}
\newacronym{sumimo}{SU-MIMO}{single-user MIMO}
\newacronym{mumimo}{MU-MIMO}{multi-user MIMO}
\newacronym{rf}{RF}{radio-frequency}
\newacronym[plural=PSNs, firstplural=phase-shift networks (PSNs)]{psn}{PSN}{phase-shifting network}
\newacronym{met}{MET}{maximum eigenmode transmission}
\newacronym[plural=DACs, firstplural=digital/analog converters (DACs)]{dac}{DAC}{digital/analog converter}
\newacronym[plural=PAs, firstplural=power amplifiers (PAs)]{pa}{PA}{power amplifier}
\newacronym{agc}{AGC}{automatic gain control}
\newacronym{svd}{SVD}{singular value decomposition}
\newacronym{ad}{A/D}{analog/digital}
\newacronym{mmwave}{mmWave}{millimeter wave}
\newacronym{cfo}{CFO}{carrier frequency offset}
\newacronym{tx}{TX}{transmitter}
\newacronym{rx}{RX}{receiver}
\newacronym{omp}{OMP}{orthogonal matching pursuit}
\newacronym{nmse}{NMSE}{normalized mean square error}
\newacronym{cs}{CS}{compressive sensing}
\newacronym{csi}{CSI}{channel state information}
\newacronym{csit}{CSIT}{channel state information at the transmitter}
\newacronym{ls}{LS}{least squares}
\newacronym{pn}{PN}{phase noise}
\newacronym{tdd}{TDD}{time division duplexing}
\newacronym{fdd}{FDD}{frequency division duplexing}
\newacronym[plural=ADCs, firstplural=analog/digital converters (ADCs)]{adc}{ADC}{analog/digital converter}
\newacronym{mer}{MER}{maximum eigenmode reception}
\newacronym{lte}{LTE}{long-term evolution}
\newacronym{bd}{BD}{block diagonalization}
\newacronym{rb}{RB}{resource block}
\newacronym{rbg}{RBG}{resource block group}
\newacronym{ofdm}{OFDM}{orthogonal frequency division multiplexing}
\newacronym{hosvd}{HOSVD}{higher-order singular value decomposition}
\newacronym{dft}{DFT}{discrete Fourier transform}
\newacronym{idft}{IDFT}{inverse discrete Fourier transform}
\newacronym{cpe}{CPE}{common phase error}
\newacronym{ici}{ICI}{intercarrier interference}
\newacronym{ura}{URA}{uniform rectangular array}
\newacronym{pdf}{pdf}{probability density function}
\newacronym{los}{LOS}{line of sight}
\newacronym{nlos}{NLOS}{non-line of sight}
\newacronym{mrt}{MRT}{maximum ratio transmission}
\newacronym{tmrt}{TMRT}{tensor maximum ratio transmission}
\newacronym{tzf}{TZF}{tensor zero-forcing}
\newacronym{tti}{TTI}{transmission time interval}
\newacronym{cdf}{cdf}{cumulative distribution function}

%% file: mathmacros.tex
\newcommand{\tran}{\mathsf{T}}						
\newcommand{\hermit}{\mathsf{H}}					

\DeclareMathOperator{\Diag}{Diag}

\DeclarePairedDelimiterX{\openinterval}[2]{]}{]}{#1,#2}
\DeclarePairedDelimiterX{\openintervalboth}[2]{]}{[}{#1,#2}
\DeclarePairedDelimiterX{\closedinterval}[2]{[}{]}{#1,#2}

%% file: macros.tex
\definecolor{raspberry}{rgb}{0.89, 0.04, 0.36}

%% file: fig2-20db-int.tex
%
%
\definecolor{mycolor1}{rgb}{0.00000,0.44700,0.74100}%
\definecolor{mycolor2}{rgb}{0.85000,0.32500,0.09800}%
\definecolor{mycolor3}{rgb}{0.92900,0.69400,0.12500}%
\definecolor{mycolor4}{rgb}{0.49400,0.18400,0.55600}%
\begin{tikzpicture}

\begin{axis}[%
width=2.5in,
height=1.5in,
scale only axis,
xmin=0,
xmax=500,
xlabel style={font=\color{white!15!black}},
xlabel={TTI index},
ymode=log,
ymin=1e-5,
ymax=1,
yminorticks=true,
ylabel style={font=\color{white!15!black}},
ylabel={Chordal distance},
axis background/.style={fill=white},
legend style={at={(0.625,0.35725)}, anchor=south west, legend cell align=left, align=left, draw=white!15!black, font=\footnotesize},
xmajorgrids,
ymajorgrids
]
\addplot [color=mycolor1]
  table[row sep=crcr]{%
1	0.0177694250643765\\
2	0.0177694250643765\\
3	0.010205376240698\\
4	0.0218775344001776\\
5	0.0336860929474383\\
6	0.0375626312546352\\
7	0.0607445125135148\\
8	0.0672371955886515\\
9	0.05907873989788\\
10	0.0636211382761537\\
11	0.0976375966436231\\
12	0.12668393307729\\
13	0.102603975923926\\
14	0.0953407890903898\\
15	0.127598162952724\\
16	0.18118361223969\\
17	0.199775644499018\\
18	0.190307783407132\\
19	0.174507665088814\\
20	0.207905213278424\\
21	0.197242454141264\\
22	0.210510416031427\\
23	0.257599247203681\\
24	0.254480096615329\\
25	0.269901470552563\\
26	0.266110470333941\\
27	0.326974354210009\\
28	0.299599572470499\\
29	0.307403046356084\\
30	0.275180480185157\\
31	0.284360933482544\\
32	0.330416871109842\\
33	0.346824693591756\\
34	0.329667470451532\\
35	0.313663553399852\\
36	0.388525720158449\\
37	0.318927509515177\\
38	0.369175151181617\\
39	0.351084084309558\\
40	0.373611165399509\\
41	0.387662276947264\\
42	0.377418872981206\\
43	0.364584037556244\\
44	0.416603364132859\\
45	0.369745948246995\\
46	0.387852953468931\\
47	0.368879486122897\\
48	0.411706457999819\\
49	0.419684930471303\\
50	0.423124896100944\\
51	0.404944521310814\\
52	0.42890228042366\\
53	0.408349979993003\\
54	0.423169150570149\\
55	0.436820618412219\\
56	0.458619191617049\\
57	0.421585691766317\\
58	0.441987602188217\\
59	0.433302470551441\\
60	0.423684851513717\\
61	0.437150828916113\\
62	0.459629446470551\\
63	0.456816609447974\\
64	0.457447475527252\\
65	0.459144765524911\\
66	0.44424261720989\\
67	0.449024068963512\\
68	0.455994340649475\\
69	0.462545532984808\\
70	0.470906115146306\\
71	0.464177833632045\\
72	0.470366069645245\\
73	0.467657108200615\\
74	0.458729840667499\\
75	0.482653659595149\\
76	0.479290873303775\\
77	0.477042798192246\\
78	0.478629340895338\\
79	0.485472783335581\\
80	0.488671705935872\\
81	0.492482410246467\\
82	0.487092124303101\\
83	0.492826375446402\\
84	0.48067116695156\\
85	0.489641984948117\\
86	0.488836708477998\\
87	0.483686955532052\\
88	0.514924411362625\\
89	0.470331577335093\\
90	0.498638979709839\\
91	0.50167051531971\\
92	0.495635301823438\\
93	0.48856934187858\\
94	0.496374232602573\\
95	0.51154504659299\\
96	0.501460960500959\\
97	0.517293897690682\\
98	0.508100947317355\\
99	0.505374776299241\\
100	0.515534077287007\\
101	0.507986154905588\\
102	0.517825816120057\\
103	0.522294139429804\\
104	0.50775758917712\\
105	0.522920516963857\\
106	0.53262904787403\\
107	0.526548727452352\\
108	0.533575691209958\\
109	0.534129778094994\\
110	0.54200638772028\\
111	0.522875796394086\\
112	0.536699502994318\\
113	0.529966803306111\\
114	0.542274470719628\\
115	0.54139458400007\\
116	0.530412230776373\\
117	0.547444916017212\\
118	0.55702731901915\\
119	0.53575848427632\\
120	0.549452669174994\\
121	0.548352561036822\\
122	0.553528540557182\\
123	0.559588247354115\\
124	0.545060020749823\\
125	0.56285255981463\\
126	0.56723740147185\\
127	0.557403905739457\\
128	0.551749128785669\\
129	0.560078907633183\\
130	0.576835468894342\\
131	0.571890390906461\\
132	0.565348630716386\\
133	0.561638863472629\\
134	0.573170587668861\\
135	0.573970501025126\\
136	0.574482894284441\\
137	0.572492877284443\\
138	0.588442880412335\\
139	0.5794713430686\\
140	0.582031221567977\\
141	0.580774512967541\\
142	0.593200351849332\\
143	0.601045426457603\\
144	0.583702553555155\\
145	0.589882002374427\\
146	0.606385400920917\\
147	0.607083982512435\\
148	0.593559763899043\\
149	0.59485598618119\\
150	0.605758538259985\\
151	0.60324106363623\\
152	0.599379213180416\\
153	0.608758877334094\\
154	0.61208878374499\\
155	0.620781090008631\\
156	0.613408237430069\\
157	0.629307252982629\\
158	0.61896258534688\\
159	0.616955487049967\\
160	0.631213907111625\\
161	0.632360400409687\\
162	0.646615432747094\\
163	0.636031372499328\\
164	0.632519488427501\\
165	0.640187368110084\\
166	0.652612843634987\\
167	0.644777811335775\\
168	0.64480647360109\\
169	0.653640426394868\\
170	0.648481262222537\\
171	0.652134297427074\\
172	0.656544359377541\\
173	0.660477244400898\\
174	0.66273148344966\\
175	0.663721624972783\\
176	0.661178930571964\\
177	0.677419212175946\\
178	0.663878404925013\\
179	0.680373551597468\\
180	0.675767938441007\\
181	0.67688206078106\\
182	0.675258960409365\\
183	0.684363518211531\\
184	0.680351955712434\\
185	0.68913201660498\\
186	0.685332075662016\\
187	0.686846449218988\\
188	0.698355049905303\\
189	0.69655244860337\\
190	0.699460510910311\\
191	0.715046431433648\\
192	0.717140396901635\\
193	0.709765805364656\\
194	0.711287123233985\\
195	0.716091130053007\\
196	0.71132690196799\\
197	0.7185888191804\\
198	0.728263624545403\\
199	0.726486467323652\\
200	0.727877557486982\\
201	0.739560559129953\\
202	0.74338787696546\\
203	0.736345626816189\\
204	0.743100166162474\\
205	0.7426011264447\\
206	0.744154209334044\\
207	0.746195854746144\\
208	0.750060747584774\\
209	0.7411991337133\\
210	0.754738826012927\\
211	0.760858794032888\\
212	0.759960686682402\\
213	0.761402774192376\\
214	0.76869986595793\\
215	0.757052588126824\\
216	0.766691566100169\\
217	0.774487719699941\\
218	0.773333350014396\\
219	0.778795953166183\\
220	0.779154941027582\\
221	0.781278187280687\\
222	0.779897987450903\\
223	0.783900528572892\\
224	0.793518130969875\\
225	0.78355649537638\\
226	0.791158467452461\\
227	0.798598726763441\\
228	0.794470720735489\\
229	0.800798992784413\\
230	0.801783067359248\\
231	0.800780266108265\\
232	0.812423645929359\\
233	0.813600707733745\\
234	0.815898669443192\\
235	0.816636750512974\\
236	0.823142142481158\\
237	0.822457954521613\\
238	0.82288152874153\\
239	0.831120832312683\\
240	0.829527192974888\\
241	0.832116280124852\\
242	0.833307253392705\\
243	0.828040640019163\\
244	0.843979340408023\\
245	0.842434824285185\\
246	0.849733585279263\\
247	0.835431999320131\\
248	0.844431041561286\\
249	0.844089342682656\\
250	0.849984687788675\\
251	0.851610795368608\\
252	0.86241949724483\\
253	0.856474140452283\\
254	0.858340937320813\\
255	0.862954047808355\\
256	0.865212080335166\\
257	0.866419133160492\\
258	0.869477668572271\\
259	0.867804814283786\\
260	0.872755239193645\\
261	0.876717141940207\\
262	0.880294814358985\\
263	0.877556209827945\\
264	0.880732422775944\\
265	0.880257190777836\\
266	0.88890563062729\\
267	0.886948003046746\\
268	0.887892232998936\\
269	0.89217444061189\\
270	0.894967678881406\\
271	0.899493915879457\\
272	0.893020119601996\\
273	0.900162346856459\\
274	0.901138255273737\\
275	0.900159577753859\\
276	0.903359017999212\\
277	0.900224840805131\\
278	0.907301018398031\\
279	0.909870548724098\\
280	0.911751550484295\\
281	0.910650988954472\\
282	0.917858088209354\\
283	0.922713425538223\\
284	0.913886223713058\\
285	0.92120995535098\\
286	0.924187539275256\\
287	0.924584167144665\\
288	0.927604031111488\\
289	0.927118294417081\\
290	0.927179137495725\\
291	0.924931910870357\\
292	0.931432370750759\\
293	0.928916935317497\\
294	0.935834046842212\\
295	0.938227993980634\\
296	0.934814093121899\\
297	0.938876413865351\\
298	0.940463931510897\\
299	0.942559907779176\\
300	0.94186839336923\\
301	0.945673696900661\\
302	0.944953425859145\\
303	0.945277726502803\\
304	0.950097178736944\\
305	0.94947116283942\\
306	0.953319169014279\\
307	0.949415529746605\\
308	0.953972866440598\\
309	0.9545598271762\\
310	0.956171468334093\\
311	0.956944644522767\\
312	0.959111248550219\\
313	0.961172533230448\\
314	0.962172277664289\\
315	0.960906647111891\\
316	0.964360908217998\\
317	0.963222691272015\\
318	0.963057689800631\\
319	0.965227706929276\\
320	0.965124986732738\\
321	0.96806423440595\\
322	0.968738562593666\\
323	0.969669155910798\\
324	0.972930457612391\\
325	0.970157990172042\\
326	0.973514915369808\\
327	0.973143963692847\\
328	0.974755951341895\\
329	0.976421828741088\\
330	0.976813708966584\\
331	0.975984372012135\\
332	0.979623854892553\\
333	0.979393655669022\\
334	0.97993707625058\\
335	0.979836480276609\\
336	0.978832601486737\\
337	0.98067521261436\\
338	0.982720756562236\\
339	0.983852897074492\\
340	0.983541450133297\\
341	0.985350745864172\\
342	0.98566990136469\\
343	0.984981021878412\\
344	0.984859146822006\\
345	0.987698359457563\\
346	0.987488929957152\\
347	0.986681609026465\\
348	0.988626034270943\\
349	0.988685878427233\\
350	0.990291365091095\\
351	0.990072108925781\\
352	0.990817013304649\\
353	0.991385106364485\\
354	0.992481682448119\\
355	0.992566627192033\\
356	0.992463345650988\\
357	0.991773368088197\\
358	0.993760687711419\\
359	0.994476456519227\\
360	0.994799768275874\\
361	0.994630573884417\\
362	0.994852453638328\\
363	0.995397379664449\\
364	0.995964782015593\\
365	0.996302257385567\\
366	0.996914773632355\\
367	0.996605345958266\\
368	0.997168952735703\\
369	0.997080153937618\\
370	0.997166847649885\\
371	0.997682658666354\\
372	0.998197845132368\\
373	0.998324427379434\\
374	0.998700154482539\\
375	0.99840640540808\\
376	0.998937506668825\\
377	0.998961181369796\\
378	0.999132149474143\\
379	0.999326207219156\\
380	0.999257446760361\\
381	0.999603848847651\\
382	0.999692585563841\\
383	0.999773502297115\\
384	0.999778331429157\\
385	0.999819131168638\\
386	0.999895440658148\\
387	0.99995180261606\\
388	0.999970049694877\\
389	0.99998106681095\\
390	0.999922006481111\\
391	0.999916942252321\\
392	0.999805944480197\\
393	0.999772451669723\\
394	0.999659225605465\\
395	0.999619597398532\\
396	0.999315846776025\\
397	0.999055189546509\\
398	0.999162115255372\\
399	0.998784000204265\\
400	0.998727975085684\\
401	0.998603486945245\\
402	0.998381297794949\\
403	0.998024011577642\\
404	0.997664064447752\\
405	0.997709439800518\\
406	0.997100925217513\\
407	0.996992139290781\\
408	0.996708791982272\\
409	0.996545976788873\\
410	0.995942208565604\\
411	0.9952263430624\\
412	0.995275583781357\\
413	0.994919156556275\\
414	0.994203984849729\\
415	0.994383386807868\\
416	0.993593764368631\\
417	0.992841792267241\\
418	0.991999014541254\\
419	0.991760086631049\\
420	0.991427997438298\\
421	0.991034382665526\\
422	0.991074056306353\\
423	0.989985697701653\\
424	0.990196915206879\\
425	0.989188268126501\\
426	0.989135110286812\\
427	0.98825734060292\\
428	0.988056934651045\\
429	0.987863298954577\\
430	0.986934003251121\\
431	0.985636234766755\\
432	0.986632165470625\\
433	0.985795323329126\\
434	0.985301276504276\\
435	0.98454657582353\\
436	0.984146094678014\\
437	0.984270993530667\\
438	0.983133192863084\\
439	0.982596705101289\\
440	0.9817894187528\\
441	0.981844262470329\\
442	0.980409990292453\\
443	0.980992257020077\\
444	0.980240449333301\\
445	0.980228624491634\\
446	0.978820827680604\\
447	0.979492151664546\\
448	0.97967939740097\\
449	0.977260723650954\\
450	0.977071796796994\\
451	0.976901409080751\\
452	0.975227459399161\\
453	0.975571000302311\\
454	0.975308733610762\\
455	0.975336863353128\\
456	0.974447179882748\\
457	0.974690310889869\\
458	0.974616267032881\\
459	0.973828149692879\\
460	0.973638803978652\\
461	0.971902877330781\\
462	0.970783375121807\\
463	0.971083804234427\\
464	0.972261599009175\\
465	0.970619323987995\\
466	0.970621331562603\\
467	0.970262927560437\\
468	0.969773444125249\\
469	0.968265257717699\\
470	0.968573450262946\\
471	0.969188204450963\\
472	0.967604032972803\\
473	0.968742847986459\\
474	0.967737899178948\\
475	0.967271432067559\\
476	0.967186116634997\\
477	0.968441897847023\\
478	0.965780679060467\\
479	0.967034678717759\\
480	0.966156793761559\\
481	0.967688973443796\\
482	0.966480760052945\\
483	0.966415831582401\\
484	0.964801597562401\\
485	0.965349074560659\\
486	0.966003430065163\\
487	0.964191527997698\\
488	0.965122936824378\\
489	0.964722912228048\\
490	0.963180462578166\\
491	0.963558776761645\\
492	0.96551379590974\\
493	0.963061825789668\\
494	0.963219569785906\\
495	0.963040569978685\\
496	0.963366832925351\\
497	0.961779608486355\\
498	0.961725536336648\\
499	0.96233721639649\\
500	0.963300315687736\\
};
\addlegendentry{Full}

\addplot [color=mycolor2]
  table[row sep=crcr]{%
1	0.0547398880025255\\
2	0.0547398880025255\\
3	0.0342034029914781\\
4	0.0545531760207637\\
5	0.0620402526233455\\
6	0.0468092222960084\\
7	0.0994454624881423\\
8	0.0962706723392316\\
9	0.0984228104179501\\
10	0.0966702311267435\\
11	0.123099558463087\\
12	0.172998690234194\\
13	0.128241271473375\\
14	0.129957770321438\\
15	0.171740328376997\\
16	0.21509308667649\\
17	0.205353817078039\\
18	0.208452277661306\\
19	0.188759368648789\\
20	0.222735852905773\\
21	0.252992214494596\\
22	0.261511358681195\\
23	0.258577169932386\\
24	0.294160121810675\\
25	0.296795336299418\\
26	0.306669206611358\\
27	0.340435271549487\\
28	0.35413089197298\\
29	0.333423046567885\\
30	0.306052756811752\\
31	0.321624325990705\\
32	0.359520967761158\\
33	0.380047775462331\\
34	0.381028916283084\\
35	0.377100938461169\\
36	0.400044351090409\\
37	0.380844425960112\\
38	0.403502453046156\\
39	0.383142688057452\\
40	0.398922841210006\\
41	0.398754007714739\\
42	0.420448037211667\\
43	0.413226656588115\\
44	0.434757195021329\\
45	0.401526707644014\\
46	0.435779016889314\\
47	0.421227821446231\\
48	0.440908423319009\\
49	0.439362644484197\\
50	0.456889323269271\\
51	0.438052660355172\\
52	0.444433897268654\\
53	0.443937934206608\\
54	0.44889506493051\\
55	0.474629643639073\\
56	0.479249383662976\\
57	0.441433730500009\\
58	0.478371902830375\\
59	0.454324543584916\\
60	0.464056977319283\\
61	0.459005808224255\\
62	0.485224465244667\\
63	0.487096129385504\\
64	0.477766977074198\\
65	0.464831122508225\\
66	0.471458563208469\\
67	0.490403831714712\\
68	0.478006687794294\\
69	0.484181898465167\\
70	0.497969498367443\\
71	0.493483958084868\\
72	0.473266081133685\\
73	0.495938130359844\\
74	0.493468613038639\\
75	0.49715859326588\\
76	0.512583223803104\\
77	0.499121786850787\\
78	0.496463720496374\\
79	0.507900846241804\\
80	0.510374341904856\\
81	0.514446138993598\\
82	0.503739852492267\\
83	0.51186144522784\\
84	0.503574398001023\\
85	0.511721556320851\\
86	0.522995219720952\\
87	0.512525249045567\\
88	0.526152652224869\\
89	0.503731877279938\\
90	0.524070544222408\\
91	0.522615224889463\\
92	0.519848964376096\\
93	0.510453752082948\\
94	0.52262291801175\\
95	0.529241291542244\\
96	0.525980478640297\\
97	0.525157331296293\\
98	0.527942328165519\\
99	0.532810106981044\\
100	0.529375583141497\\
101	0.52743001123937\\
102	0.53741789256905\\
103	0.539402073113636\\
104	0.537518492683758\\
105	0.556458694283654\\
106	0.548314719355732\\
107	0.540150614848738\\
108	0.551844368464692\\
109	0.545703338774233\\
110	0.552762168034038\\
111	0.555737195619016\\
112	0.550815855740302\\
113	0.549346647139181\\
114	0.560650217859578\\
115	0.559859690270518\\
116	0.54882728677965\\
117	0.566739823368389\\
118	0.567718021519618\\
119	0.562807193744277\\
120	0.564939868818207\\
121	0.558995229142721\\
122	0.57273450946229\\
123	0.579373960568965\\
124	0.560931432157\\
125	0.561765982789228\\
126	0.573825260610251\\
127	0.575698875593466\\
128	0.572075176601457\\
129	0.574370328256963\\
130	0.591236056782087\\
131	0.582016216629519\\
132	0.580500404514723\\
133	0.574551496796628\\
134	0.585849817533906\\
135	0.585721317004129\\
136	0.58634406302001\\
137	0.58606346349763\\
138	0.596532130226403\\
139	0.599125454622828\\
140	0.594270943755603\\
141	0.598363646045153\\
142	0.597945624468823\\
143	0.611528918569634\\
144	0.602663129330036\\
145	0.607402853428192\\
146	0.611829764337243\\
147	0.615257664908615\\
148	0.604352039334822\\
149	0.609895736491319\\
150	0.615736765278088\\
151	0.62051114041834\\
152	0.616639276365701\\
153	0.622243716716639\\
154	0.625100819375064\\
155	0.623968262506966\\
156	0.632664420681811\\
157	0.633413648425775\\
158	0.629999275596375\\
159	0.629877759413839\\
160	0.639286846881516\\
161	0.641382660419632\\
162	0.652247688902601\\
163	0.646050862079672\\
164	0.641081153543822\\
165	0.652840359926144\\
166	0.658201855064681\\
167	0.654524410435781\\
168	0.656256306139345\\
169	0.660003468277004\\
170	0.662886684864162\\
171	0.659405902554674\\
172	0.663630173969942\\
173	0.670412210492216\\
174	0.668304620759404\\
175	0.67321798045077\\
176	0.672561403666628\\
177	0.68204426221173\\
178	0.675426604336644\\
179	0.683821534789964\\
180	0.682203140884382\\
181	0.680917480109002\\
182	0.685121799674142\\
183	0.690303247903507\\
184	0.693711721245542\\
185	0.690429338458654\\
186	0.699449312242429\\
187	0.696839551737385\\
188	0.704781739733263\\
189	0.708645727473245\\
190	0.713263147120381\\
191	0.718479278194729\\
192	0.718895634770524\\
193	0.711475811451957\\
194	0.714922726980253\\
195	0.724979371199748\\
196	0.723318458953429\\
197	0.729489396208881\\
198	0.729579553789546\\
199	0.728036941588877\\
200	0.737646205818421\\
201	0.735268440498264\\
202	0.739315568752019\\
203	0.742721000697905\\
204	0.746791950804606\\
205	0.740128705317261\\
206	0.748352452028701\\
207	0.752172022308278\\
208	0.761733975909586\\
209	0.751011341864689\\
210	0.761177016348236\\
211	0.763792900917992\\
212	0.761904804704269\\
213	0.764146088632309\\
214	0.772416642285208\\
215	0.768255062220638\\
216	0.769527415483745\\
217	0.777555106143211\\
218	0.776282813175392\\
219	0.781009378368596\\
220	0.777139414795215\\
221	0.783369468232131\\
222	0.786398546408064\\
223	0.787046537060602\\
224	0.790804800380296\\
225	0.792214527529635\\
226	0.796522941984114\\
227	0.799736049620406\\
228	0.798657755427532\\
229	0.801841532375494\\
230	0.805619597149694\\
231	0.80516351295252\\
232	0.813676562334768\\
233	0.814900341294921\\
234	0.817685590844765\\
235	0.812556162371288\\
236	0.822180996708174\\
237	0.822782916159826\\
238	0.825632132611707\\
239	0.829118036818667\\
240	0.830362165223028\\
241	0.834137838531955\\
242	0.836744911699069\\
243	0.831995317403225\\
244	0.845503536417455\\
245	0.841910306936821\\
246	0.847083219731953\\
247	0.836137089025376\\
248	0.846898005251576\\
249	0.848108446424843\\
250	0.848920490833596\\
251	0.850852842982723\\
252	0.859143236723767\\
253	0.858375333784126\\
254	0.85842962300845\\
255	0.865743259088515\\
256	0.864701361728827\\
257	0.867867717681743\\
258	0.869371892299071\\
259	0.870052587136114\\
260	0.872866304552962\\
261	0.874847757129883\\
262	0.881242870394576\\
263	0.875898778569456\\
264	0.88294140151501\\
265	0.88000189318136\\
266	0.888262419115587\\
267	0.886743280931759\\
268	0.887881172644435\\
269	0.891551211591091\\
270	0.895253800256986\\
271	0.896167859810031\\
272	0.895587882759497\\
273	0.897495454091542\\
274	0.901612664970829\\
275	0.898841125890341\\
276	0.902645983710278\\
277	0.903709822419136\\
278	0.907437066144894\\
279	0.910480070316773\\
280	0.913855361559892\\
281	0.911700979500717\\
282	0.913443529580615\\
283	0.918589192103516\\
284	0.916188880990026\\
285	0.922601227283177\\
286	0.921389840770606\\
287	0.923367250743866\\
288	0.926896896203311\\
289	0.924696310784972\\
290	0.925275520205677\\
291	0.92847924092312\\
292	0.931167364633454\\
293	0.928194431276808\\
294	0.934014482131473\\
295	0.93609279591747\\
296	0.935036097490887\\
297	0.937107725859862\\
298	0.939541132974054\\
299	0.942768361255075\\
300	0.941068493829476\\
301	0.945676101004967\\
302	0.946114371960866\\
303	0.946367328674915\\
304	0.950072132503241\\
305	0.950694461069248\\
306	0.950393122599323\\
307	0.950982942679498\\
308	0.95428362742414\\
309	0.955115419628264\\
310	0.954269048918227\\
311	0.956648539148814\\
312	0.957568103987273\\
313	0.960771261933907\\
314	0.960197365852234\\
315	0.962251123373526\\
316	0.963019243351065\\
317	0.961034049001271\\
318	0.962457341711817\\
319	0.964501244466424\\
320	0.96406771658763\\
321	0.967102703763106\\
322	0.968038586527147\\
323	0.967971980302416\\
324	0.970923642449259\\
325	0.970572909101682\\
326	0.971025641375232\\
327	0.972279579364934\\
328	0.97396782514174\\
329	0.9742933899782\\
330	0.975512281941478\\
331	0.976489975092111\\
332	0.977688742825582\\
333	0.977691926946381\\
334	0.979274682102944\\
335	0.978986205664613\\
336	0.979404047936151\\
337	0.981938123998127\\
338	0.981069098236391\\
339	0.982840801953007\\
340	0.983733897406587\\
341	0.98513667598881\\
342	0.985366784631732\\
343	0.985431771692432\\
344	0.984984225315225\\
345	0.987118725429133\\
346	0.986772323377686\\
347	0.987120528188107\\
348	0.988830119381378\\
349	0.988462834849404\\
350	0.990719962814277\\
351	0.989419654461983\\
352	0.98984759023683\\
353	0.991094652166284\\
354	0.991426938804074\\
355	0.991914268125509\\
356	0.992308627509803\\
357	0.991501353260382\\
358	0.993359353716485\\
359	0.994286362424256\\
360	0.994417896021304\\
361	0.994434848011714\\
362	0.994780695764916\\
363	0.995310600380034\\
364	0.995763293846778\\
365	0.996136323427488\\
366	0.996548989060355\\
367	0.996326111378314\\
368	0.99720987201258\\
369	0.996811644576824\\
370	0.997146263404131\\
371	0.997590867639795\\
372	0.998118560896512\\
373	0.998195844627904\\
374	0.998674902471651\\
375	0.998352794202692\\
376	0.998847190630834\\
377	0.998878694226423\\
378	0.999156219320999\\
379	0.999249560553252\\
380	0.999172305725504\\
381	0.999531908697699\\
382	0.999652566814549\\
383	0.999731667564975\\
384	0.999726079935778\\
385	0.999782778542183\\
386	0.999877591427764\\
387	0.999934595686749\\
388	0.999965430108743\\
389	0.999978806714488\\
390	0.999938813256556\\
391	0.999940480768127\\
392	0.999823638581381\\
393	0.999774564188769\\
394	0.999716098045104\\
395	0.999668036017557\\
396	0.999385955831365\\
397	0.999156925104282\\
398	0.999283327529496\\
399	0.998881349725309\\
400	0.998809358541755\\
401	0.998691552878622\\
402	0.998544283232293\\
403	0.998267551402161\\
404	0.997858813725403\\
405	0.997931066472571\\
406	0.997226895522626\\
407	0.997213253398772\\
408	0.996855335713592\\
409	0.996803796290295\\
410	0.99618611597273\\
411	0.995587713681311\\
412	0.995143590180311\\
413	0.995137634841914\\
414	0.99458945528362\\
415	0.994505035966734\\
416	0.99406377405114\\
417	0.993111602445863\\
418	0.992322269172262\\
419	0.992040198309513\\
420	0.991655938171418\\
421	0.991592942754655\\
422	0.991521185725719\\
423	0.99060010958611\\
424	0.9907313485275\\
425	0.989438584819776\\
426	0.989392941747831\\
427	0.988880020750796\\
428	0.988142383985323\\
429	0.988185255859637\\
430	0.987537244056575\\
431	0.986662186763994\\
432	0.987373792953848\\
433	0.98616924257512\\
434	0.985819917195875\\
435	0.984782412510357\\
436	0.984981007032905\\
437	0.984315833952808\\
438	0.984238472631856\\
439	0.983292744837432\\
440	0.982755302758954\\
441	0.98268629401053\\
442	0.981495591858038\\
443	0.982106390507663\\
444	0.980529627853394\\
445	0.981308709618303\\
446	0.979256892758264\\
447	0.980236623925906\\
448	0.97929172420074\\
449	0.978090075085223\\
450	0.977880918958018\\
451	0.978038794781012\\
452	0.976657881678849\\
453	0.975766171724415\\
454	0.976541895624803\\
455	0.976144238935998\\
456	0.976269405119765\\
457	0.97458865217676\\
458	0.975463456636898\\
459	0.975726290906196\\
460	0.974821772116278\\
461	0.972751897824539\\
462	0.973086333398701\\
463	0.972200341492133\\
464	0.973455903417522\\
465	0.971859438195232\\
466	0.972362632764419\\
467	0.971614756292701\\
468	0.971419413119659\\
469	0.971101491929455\\
470	0.969710759590329\\
471	0.969708253678374\\
472	0.96997914706135\\
473	0.970331015085967\\
474	0.968972032286284\\
475	0.967523543842499\\
476	0.968031304955853\\
477	0.968924132469255\\
478	0.967369993144441\\
479	0.967284150194227\\
480	0.96769163061746\\
481	0.968492212687351\\
482	0.967726422438015\\
483	0.967450654686461\\
484	0.965695485294304\\
485	0.965865513876209\\
486	0.966575138202102\\
487	0.965217911676869\\
488	0.966170775740777\\
489	0.96500791820033\\
490	0.964389599905648\\
491	0.964644747439314\\
492	0.965965459104189\\
493	0.964552218133049\\
494	0.964689799211779\\
495	0.96455756803707\\
496	0.963859729362569\\
497	0.96356532953274\\
498	0.96356835184373\\
499	0.963614094054141\\
500	0.966057345988601\\
};
\addlegendentry{Horizontal}

\addplot [color=mycolor3]
  table[row sep=crcr]{%
1	0.000135308310668947\\
2	0.000135308310668947\\
3	0.000201170111287319\\
4	0.000131342476832097\\
5	0.000206516206439766\\
6	0.000141574032125613\\
7	0.000157046427223917\\
8	0.000123678911624647\\
9	0.000113477840268406\\
10	0.000158816444430721\\
11	0.000135061043746842\\
12	0.000172912490471477\\
13	0.000202196001027366\\
14	0.000129510645943276\\
15	0.000116257436410727\\
16	0.000104378851608289\\
17	0.000229346807679542\\
18	0.000192900810833396\\
19	0.000105016279071013\\
20	0.000183225741296678\\
21	0.000124767273350135\\
22	0.000120148252134722\\
23	0.000244423294296736\\
24	0.000168651238713802\\
25	8.43025256203056e-05\\
26	0.000134887893885283\\
27	0.000148093036040264\\
28	0.000116824992597764\\
29	8.4701565343337e-05\\
30	0.000224348356344939\\
31	0.000226122004412566\\
32	0.000120835378329054\\
33	0.000123538431706116\\
34	0.000153533527033911\\
35	0.000344852114178162\\
36	0.000210724514586569\\
37	0.000105177070162188\\
38	0.000168730544166673\\
39	0.000226528829457295\\
40	0.000184233871385509\\
41	7.26117169614127e-05\\
42	0.000366623447943015\\
43	0.000128785626374806\\
44	0.000164599328650172\\
45	0.000101977297031797\\
46	0.000223469750391569\\
47	8.01990845289424e-05\\
48	0.000183104780774068\\
49	0.000147208024506673\\
50	0.000149664854410536\\
51	0.000118368475823483\\
52	0.00013091793155019\\
53	0.000222639055203311\\
54	0.00010064964148\\
55	0.000149807414184833\\
56	0.000145325378536054\\
57	0.000149207105290794\\
58	0.000146969474421021\\
59	0.000269960517255363\\
60	0.000116347183108512\\
61	0.00030394793885008\\
62	0.000159186502943864\\
63	0.000120197817805678\\
64	0.000135802952360309\\
65	0.000260910709395579\\
66	0.000132518088354272\\
67	0.00010352249936878\\
68	0.000101620175507589\\
69	0.000169768597453923\\
70	8.37283340218487e-05\\
71	0.000111655387301568\\
72	0.000161497423739698\\
73	0.000114995491934278\\
74	0.000103255065577013\\
75	0.000152208154983324\\
76	0.000184982482990359\\
77	0.000169320072112644\\
78	0.000193169688837336\\
79	8.01838308358893e-05\\
80	0.00019468707829613\\
81	0.000262188591780932\\
82	0.000113924279546085\\
83	0.000131373497053933\\
84	0.00014500529714353\\
85	0.000116712407133901\\
86	0.000147918790106494\\
87	8.42862548318224e-05\\
88	0.00015029350582404\\
89	0.000141129014500552\\
90	0.000133866711427155\\
91	0.00012209024147487\\
92	0.000160716062123545\\
93	0.000153940103304639\\
94	0.000212894895113536\\
95	0.000162028952366522\\
96	8.30842537878795e-05\\
97	0.000149528197821724\\
98	0.000111028514888545\\
99	0.000212609216313486\\
100	0.000115741511421852\\
101	0.000149675646564174\\
102	0.000162214651893056\\
103	0.000100944758508481\\
104	0.000112769900310972\\
105	0.000139989434798282\\
106	0.000195785178817687\\
107	0.000129332507094693\\
108	0.000129461173797385\\
109	0.000203141978887778\\
110	0.000165095364119217\\
111	0.000194213653975617\\
112	0.000152802196705859\\
113	0.000196113088262073\\
114	0.000108256059798617\\
115	0.000188205601646785\\
116	0.000175820309105357\\
117	0.000141584332390632\\
118	0.000147774453339122\\
119	0.000199033306952234\\
120	6.48361372870943e-05\\
121	7.04306771845298e-05\\
122	0.000159808545430351\\
123	8.32789818572266e-05\\
124	0.000151356570655004\\
125	0.000155017204302155\\
126	0.000157253071829\\
127	0.000170683054602627\\
128	0.000143525028558211\\
129	0.000121269895501674\\
130	9.29022434781102e-05\\
131	0.000200689390522235\\
132	0.000185180399470342\\
133	0.000187015719041184\\
134	0.000173747174789018\\
135	0.000160911120605822\\
136	0.000126442965916618\\
137	0.000114330412427954\\
138	0.000162195680216271\\
139	0.000113915560980571\\
140	0.000131947886146999\\
141	0.000165202999193814\\
142	0.000188399478364787\\
143	0.000258387702102347\\
144	0.000133041791407851\\
145	0.000125725338546434\\
146	0.000123109650400999\\
147	0.000115877697566991\\
148	0.000206549224109498\\
149	0.000164093976500823\\
150	0.000201987592786879\\
151	0.00015875899282275\\
152	0.000265506085791023\\
153	0.000124325185817908\\
154	8.62077834283936e-05\\
155	0.000203290599365347\\
156	0.000164210439044277\\
157	0.00015028059990736\\
158	0.00013010240136796\\
159	0.000277654719548659\\
160	0.000164837488363501\\
161	0.000138903951207947\\
162	0.000118310239005337\\
163	0.000185765871863608\\
164	0.000140401299643889\\
165	0.000112673032979027\\
166	0.00013927939404581\\
167	0.000133092039668625\\
168	0.000225837023870445\\
169	0.000172028106582556\\
170	0.000157839576932406\\
171	0.000184139815033835\\
172	0.000158916289937971\\
173	0.000203986535173672\\
174	0.000180568687743876\\
175	0.000116234647040625\\
176	0.000132910635438421\\
177	0.000119710744194779\\
178	0.000281338254348185\\
179	9.82490664095503e-05\\
180	0.00012391593919735\\
181	0.000166311813329012\\
182	0.000258557432393014\\
183	0.000297186161928087\\
184	0.000167228066066238\\
185	0.000127793699047829\\
186	0.000284319730363525\\
187	0.000162374605262727\\
188	0.000189770716218274\\
189	0.000240684107392897\\
190	0.000141835217845443\\
191	0.000213183024088381\\
192	0.000162130420426743\\
193	0.00022132168846718\\
194	0.000221163527060053\\
195	0.000124900536113448\\
196	0.000135902586858894\\
197	8.03246881078024e-05\\
198	0.000147090396634397\\
199	9.22374609632071e-05\\
200	0.000205034714663566\\
201	0.000242874162632989\\
202	0.000144943957920396\\
203	0.000252855491391646\\
204	0.00013785497892963\\
205	0.000166093707459081\\
206	0.000159443326157971\\
207	0.000229724432888856\\
208	8.4993512774878e-05\\
209	0.000140364887886824\\
210	0.000139597581616879\\
211	8.43744544624325e-05\\
212	0.000123293324267293\\
213	0.000192913474399026\\
214	0.000279129831291214\\
215	0.000124291678268396\\
216	0.000133900821952693\\
217	0.000130364090614876\\
218	0.000114115903829282\\
219	0.000168027292762907\\
220	0.000166087611139576\\
221	0.000161643104625431\\
222	0.000177199835785258\\
223	0.000106624485780782\\
224	0.000178113564905025\\
225	0.000175831682480487\\
226	0.00015480039570126\\
227	0.00017508248165824\\
228	9.69287068787295e-05\\
229	0.000222965742189518\\
230	0.00016450183603588\\
231	9.67060631313021e-05\\
232	0.000235443266029745\\
233	7.54203775429452e-05\\
234	0.000184652617496983\\
235	0.000120602137274239\\
236	0.000124125162223387\\
237	0.000148883882993223\\
238	0.000181081846343811\\
239	9.40383802717326e-05\\
240	0.000156749050949778\\
241	0.000131272594749221\\
242	0.000110759184458087\\
243	0.000183232585795079\\
244	0.000203014539174429\\
245	0.000261436966499706\\
246	0.000202725573404905\\
247	0.000188224128493797\\
248	0.000155586693871579\\
249	0.000287470807684331\\
250	0.000132144664387412\\
251	0.000241331506223486\\
252	0.000139011667344224\\
253	0.000153729601109998\\
254	0.000266211904952873\\
255	8.80849765226954e-05\\
256	0.000134580810926765\\
257	0.000123412219152041\\
258	0.000184630861659563\\
259	0.000217132962539135\\
260	7.92334464587929e-05\\
261	0.000119731751746244\\
262	0.000116960309229519\\
263	0.000211605138907589\\
264	0.000100483692011499\\
265	0.000138575295120558\\
266	0.000203576482258172\\
267	0.000191919769473536\\
268	0.000275347040606733\\
269	0.00020980372291981\\
270	0.000142989071371635\\
271	7.7794666768094e-05\\
272	0.000222584988589081\\
273	0.000102827318823029\\
274	0.000146002454297978\\
275	0.000117420951025282\\
276	0.000151126415102332\\
277	0.000277922949725129\\
278	0.000243372371218431\\
279	0.000235530401592776\\
280	0.000118387252141627\\
281	0.000191682195524912\\
282	0.000131859800771728\\
283	0.000158927869676129\\
284	0.000152158678776071\\
285	0.000141149714105171\\
286	0.000122884464615969\\
287	0.000185699920603077\\
288	0.000109292643204029\\
289	0.000275949171623013\\
290	0.000199003152013422\\
291	0.000166504147699698\\
292	0.000144334084773767\\
293	9.78865178073995e-05\\
294	8.99496215986728e-05\\
295	0.000286590626871275\\
296	0.000227394313035456\\
297	0.000223077487887136\\
298	0.000132306171609431\\
299	0.0001169367009346\\
300	0.000158740371001287\\
301	0.000123981605833567\\
302	0.000118143167066298\\
303	0.000226481805571288\\
304	0.000107307074953344\\
305	0.000116065744022575\\
306	8.39198422072029e-05\\
307	0.000275617139691153\\
308	0.000106510911018132\\
309	0.000136732515327054\\
310	0.000289861778050393\\
311	0.000173892269386688\\
312	0.000148540332775826\\
313	0.000155956298503401\\
314	0.000159330393709156\\
315	0.000177566940400442\\
316	6.96112454813158e-05\\
317	0.000133443461556026\\
318	0.000277596637184263\\
319	0.000147544960710488\\
320	0.00024792467956124\\
321	0.000115399204465239\\
322	0.000118881358221945\\
323	0.000234064100087128\\
324	0.000132639750945029\\
325	0.000110353480377201\\
326	0.000189937893596615\\
327	0.000152505524375013\\
328	0.000131129795584473\\
329	0.000119565131114485\\
330	0.000135842081531812\\
331	0.000174581705569954\\
332	0.000151202212053314\\
333	0.000175475196048869\\
334	9.69812019301131e-05\\
335	0.000240482205397352\\
336	0.000128250364162563\\
337	0.000155963489865685\\
338	9.97310005772523e-05\\
339	0.000188455023516065\\
340	0.000156665461028682\\
341	0.000177643717173204\\
342	0.000101804346768308\\
343	0.000129234501774245\\
344	0.000205919582102698\\
345	0.000183580879527956\\
346	0.000196406355687728\\
347	0.000129472397601937\\
348	0.00011469500580491\\
349	0.00015851326815931\\
350	0.000103960799384428\\
351	0.000272515417476871\\
352	9.65171085389627e-05\\
353	0.000164660840747421\\
354	0.000100844775341158\\
355	0.000193266578354145\\
356	0.00012375252664677\\
357	0.000220356821773249\\
358	0.000141880886159473\\
359	9.70444998629194e-05\\
360	0.000167527353628549\\
361	0.000171397888570313\\
362	0.00017982007439824\\
363	0.000106971945394996\\
364	0.000174931420599889\\
365	0.000185660706438884\\
366	0.000148361204972036\\
367	0.000149357904581027\\
368	9.15650035247073e-05\\
369	0.000128682822186998\\
370	0.000157698966078779\\
371	0.000189235104830432\\
372	0.000145206827548527\\
373	0.000152374570221037\\
374	0.000141294427656835\\
375	0.000249195361644006\\
376	0.000126969935625432\\
377	0.000203718351909921\\
378	0.000112271147658849\\
379	0.000192750895053195\\
380	0.000149230116136201\\
381	0.00017166891590501\\
382	0.000110707683782463\\
383	0.000109068959303904\\
384	0.000198315952851358\\
385	0.000252792731416562\\
386	0.00017069528623066\\
387	0.000143258574556904\\
388	0.000159071519913001\\
389	0.000113662085643884\\
390	0.000114906228314104\\
391	0.000142415987190836\\
392	9.71073823264001e-05\\
393	0.000161975523700331\\
394	7.00024605857275e-05\\
395	0.00012587101164494\\
396	0.000247771268921604\\
397	0.000125253381148005\\
398	0.000316952660649017\\
399	0.000183544021086968\\
400	0.000292691033194115\\
401	0.000159194566854137\\
402	0.000255337129716082\\
403	0.000210873410022872\\
404	0.0001662805001451\\
405	0.000159353779848748\\
406	7.11058254394214e-05\\
407	0.000116564328501212\\
408	0.000124092156597111\\
409	0.000186423472117669\\
410	0.00012075110403188\\
411	8.57367675328402e-05\\
412	0.000220883001962857\\
413	0.000103588491486684\\
414	0.00017805759991304\\
415	0.000145091537971176\\
416	0.00014694809729846\\
417	0.000150384399942105\\
418	0.000136219648919667\\
419	0.000116348936737387\\
420	0.000166673132012074\\
421	0.000138079458499607\\
422	0.000227439301122379\\
423	0.000101307064215916\\
424	0.000135225913798154\\
425	0.000166522695605253\\
426	0.000240139511576798\\
427	0.000143556515929288\\
428	8.21274541033423e-05\\
429	0.000239359568445796\\
430	0.000168301747124566\\
431	0.000111787471069225\\
432	0.000149547433537967\\
433	0.000143292001600648\\
434	0.000232065655863478\\
435	0.000133230655056182\\
436	0.000101081361958466\\
437	0.000203201057626423\\
438	0.000162321164020673\\
439	0.000153937012074257\\
440	0.000111110002764847\\
441	0.000163163615297102\\
442	0.000164898567532723\\
443	8.47828173874077e-05\\
444	0.000122199247796595\\
445	7.45768030815941e-05\\
446	0.000153482862746979\\
447	0.00018662117562408\\
448	0.00013613371516239\\
449	7.88524462717377e-05\\
450	0.00024183367000562\\
451	0.000120193597358054\\
452	0.000107861724388936\\
453	0.000247717247827661\\
454	0.000315090100514182\\
455	0.000146695242924233\\
456	0.000112210897124077\\
457	0.000187922585945144\\
458	0.000127011287064616\\
459	0.000192550203744546\\
460	0.000189988457414858\\
461	0.000121740484011268\\
462	0.000242324543540029\\
463	0.000167851675484909\\
464	0.000125168284622634\\
465	0.000108355292370566\\
466	0.000165686858734837\\
467	0.000152180808263425\\
468	0.000114953397702156\\
469	0.000146942232609737\\
470	0.00014988980429449\\
471	0.000159178416813222\\
472	7.89651940832803e-05\\
473	0.000151528371622711\\
474	0.000259491868516704\\
475	0.00010580084237255\\
476	0.000142061736170351\\
477	0.000130583265500883\\
478	0.000120045156142368\\
479	0.000173480859046793\\
480	0.000200015754213179\\
481	0.000169363503870745\\
482	8.52743466722594e-05\\
483	0.000142078694785885\\
484	0.000110078339569164\\
485	0.000164000695022759\\
486	0.000147897206862246\\
487	0.000167952323239251\\
488	0.000205334918267408\\
489	0.00011284222335709\\
490	0.000170784917851796\\
491	0.000128996058530517\\
492	0.000190073047103356\\
493	0.000153330096406212\\
494	0.00017271728615873\\
495	0.000152372801737821\\
496	0.00025352918221162\\
497	0.00011857405701251\\
498	8.79643771568905e-05\\
499	0.000138114515048171\\
500	0.000178778963476889\\
};
\addlegendentry{Vertical}

\end{axis}
\end{tikzpicture}%

%% file: fig5-zf-runtime.tex
%
%
\definecolor{mycolor1}{rgb}{0.00000,0.44700,0.74100}%
\definecolor{mycolor2}{rgb}{0.85000,0.32500,0.09800}%
\begin{tikzpicture}

\begin{axis}[%
width=2.5in,
height=1.5in,
scale only axis,
xmin=0,
xmax=50,
xlabel style={font=\color{white!15!black}},
xlabel={Runtime $t$ [ms]},
ymin=0,
ymax=1,
ylabel style={font=\color{white!15!black}},
ylabel={$F(t)$},
axis background/.style={fill=white},
xmajorgrids,
ymajorgrids,
legend style={at={(0.95,0.4)}, legend cell align=left, align=left, draw=white!15!black, font=\footnotesize}
]
\addplot [color=mycolor1, line width=1.5pt]
  table[row sep=crcr]{%
23.027	0\\
23.027	0.01\\
23.3521	0.02\\
23.4875	0.03\\
23.8573	0.04\\
24.0316	0.05\\
24.2001	0.0600000000000001\\
24.225	0.0700000000000001\\
24.341	0.0800000000000001\\
24.5103	0.0900000000000001\\
24.6385	0.1\\
24.7728	0.11\\
24.7864	0.12\\
24.8367	0.13\\
24.9045	0.14\\
24.9808	0.15\\
24.9872	0.16\\
25.0159	0.17\\
25.0779	0.18\\
25.1278	0.19\\
25.1316	0.2\\
25.144	0.21\\
25.2045	0.22\\
25.2675	0.23\\
25.2739	0.24\\
25.3256	0.25\\
25.3599	0.26\\
25.3652	0.27\\
25.3793	0.28\\
25.4054	0.29\\
25.4342	0.3\\
25.5008	0.31\\
25.5256	0.32\\
25.5286	0.33\\
25.5339	0.34\\
25.5733	0.35\\
25.5812	0.36\\
25.5999	0.37\\
25.6353	0.38\\
25.6783	0.39\\
25.7562	0.4\\
25.7982	0.41\\
25.8275	0.42\\
25.8467	0.43\\
25.8974	0.44\\
25.9321	0.45\\
25.9483	0.46\\
26.043	0.47\\
26.085	0.48\\
26.0853	0.49\\
26.1541	0.5\\
26.1696	0.51\\
26.2251	0.52\\
26.2394	0.53\\
26.3587	0.54\\
26.375	0.55\\
26.4202	0.56\\
26.4666	0.570000000000001\\
26.5053	0.580000000000001\\
26.5795	0.590000000000001\\
26.6004	0.600000000000001\\
26.6159	0.610000000000001\\
26.6281	0.620000000000001\\
26.6697	0.63\\
26.7184	0.64\\
26.7931	0.65\\
26.83	0.66\\
26.9086	0.67\\
26.971	0.68\\
27.0166	0.69\\
27.0521	0.7\\
27.0806	0.71\\
27.2737	0.72\\
27.2833	0.73\\
27.2984	0.74\\
27.3526	0.75\\
27.4209	0.76\\
27.4514	0.77\\
27.8673	0.78\\
27.967	0.79\\
28.0135	0.8\\
28.0218	0.81\\
28.2511	0.82\\
28.33	0.83\\
28.3345	0.84\\
28.3934	0.85\\
28.609	0.86\\
28.829	0.87\\
28.849	0.88\\
29.0288	0.89\\
29.1484	0.9\\
29.2264	0.91\\
29.5327	0.92\\
29.6743	0.93\\
29.6851	0.94\\
29.7361	0.95\\
30.0278	0.96\\
30.1474	0.97\\
30.6559	0.98\\
30.6681	0.99\\
41.8641	1\\
};
\addlegendentry{ZF}

\addplot [color=mycolor2, line width=1.5pt]
  table[row sep=crcr]{%
10.1333	0\\
10.1333	0.01\\
10.2332	0.02\\
10.2528	0.03\\
10.2564	0.04\\
10.257	0.05\\
10.259	0.0600000000000001\\
10.2615	0.0700000000000001\\
10.2917	0.0800000000000001\\
10.293	0.0900000000000001\\
10.3285	0.1\\
10.3378	0.11\\
10.3394	0.12\\
10.3502	0.13\\
10.3646	0.14\\
10.3692	0.15\\
10.3755	0.16\\
10.3816	0.17\\
10.3981	0.18\\
10.3999	0.19\\
10.4142	0.2\\
10.4148	0.21\\
10.4229	0.22\\
10.4439	0.23\\
10.4701	0.24\\
10.491	0.25\\
10.5129	0.26\\
10.5416	0.27\\
10.5557	0.28\\
10.5811	0.29\\
10.5815	0.3\\
10.5955	0.31\\
10.6019	0.32\\
10.6019	0.33\\
10.6364	0.34\\
10.6368	0.35\\
10.6449	0.36\\
10.6451	0.37\\
10.6584	0.38\\
10.6726	0.39\\
10.6769	0.4\\
10.699	0.41\\
10.7217	0.42\\
10.7307	0.43\\
10.7335	0.44\\
10.7555	0.45\\
10.7942	0.46\\
10.7966	0.47\\
10.8467	0.48\\
10.8542	0.49\\
10.8656	0.5\\
10.8713	0.51\\
10.8906	0.52\\
10.9229	0.53\\
10.9336	0.54\\
10.9601	0.55\\
10.9958	0.56\\
11.0341	0.570000000000001\\
11.0409	0.580000000000001\\
11.0625	0.590000000000001\\
11.1122	0.600000000000001\\
11.1192	0.610000000000001\\
11.1465	0.620000000000001\\
11.1698	0.63\\
11.1828	0.64\\
11.212	0.65\\
11.2549	0.66\\
11.2718	0.67\\
11.2729	0.68\\
11.2852	0.69\\
11.2856	0.7\\
11.3346	0.71\\
11.3632	0.72\\
11.4466	0.73\\
11.4851	0.74\\
11.5943	0.75\\
11.6166	0.76\\
11.6751	0.77\\
11.7183	0.78\\
11.8078	0.79\\
11.9249	0.8\\
12.0504	0.81\\
12.188	0.82\\
12.2665	0.83\\
12.2687	0.84\\
12.3629	0.85\\
12.3858	0.86\\
12.5373	0.87\\
12.797	0.88\\
12.8322	0.89\\
12.8927	0.9\\
13.1009	0.91\\
13.1891	0.92\\
13.2133	0.93\\
13.2693	0.94\\
13.3166	0.95\\
13.6083	0.96\\
14.1482	0.97\\
14.2205	0.98\\
15.119	0.99\\
33.3159	1\\
};
\addlegendentry{TZF (H and \\ V updated)}

\end{axis}
\end{tikzpicture}%

%% file: draft.bbl
\begin{thebibliography}{10}
\providecommand{\url}[1]{#1}
\csname url@samestyle\endcsname
\providecommand{\newblock}{\relax}
\providecommand{\bibinfo}[2]{#2}
\providecommand{\BIBentrySTDinterwordspacing}{\spaceskip=0pt\relax}
\providecommand{\BIBentryALTinterwordstretchfactor}{4}
\providecommand{\BIBentryALTinterwordspacing}{\spaceskip=\fontdimen2\font plus
\BIBentryALTinterwordstretchfactor\fontdimen3\font minus
  \fontdimen4\font\relax}
\providecommand{\BIBforeignlanguage}[2]{{%
\expandafter\ifx\csname l@#1\endcsname\relax
\typeout{** WARNING: IEEEtran.bst: No hyphenation pattern has been}%
\typeout{** loaded for the language `#1'. Using the pattern for}%
\typeout{** the default language instead.}%
\else
\language=\csname l@#1\endcsname
\fi
#2}}
\providecommand{\BIBdecl}{\relax}
\BIBdecl

\bibitem{marzetta_fundamentals_2016}
T.~L. Marzetta, \emph{Fundamentals of massive {MIMO}}.\hskip 1em plus 0.5em
  minus 0.4em\relax Cambridge University Press, 2016.

\bibitem{van_der_perre_efficient_2018}
L.~Van~der Perre, L.~Liu, and E.~G. Larsson, ``Efficient {DSP} and {Circuit}
  {Architectures} for {Massive} {MIMO}: {State} of the {Art} and {Future}
  {Directions},'' \emph{IEEE Transactions on Signal Processing}, vol.~66,
  no.~18, pp. 4717--4736, Sep. 2018.

\bibitem{ying_kronecker_2014}
D.~Ying, F.~W. Vook, T.~A. Thomas, D.~J. Love, and A.~Ghosh, ``Kronecker
  product correlation model and limited feedback codebook design in a {3D}
  channel model,'' in \emph{Proc. {IEEE} {ICC}'14}, Sydney, NSW, Jun. 2014, pp.
  5865--5870.

\bibitem{qiu_downlink_2018}
S.~Qiu, {Da Chen}, D.~Qu, K.~Luo, and T.~Jiang, ``Downlink {Precoding} {With}
  {Mixed} {Statistical} and {Imperfect} {Instantaneous} {CSI} for {Massive}
  {MIMO} {Systems},'' \emph{IEEE Transactions on Vehicular Technology},
  vol.~67, no.~4, pp. 3028--3041, Apr. 2018.

\bibitem{schwarz_robust_2018}
S.~Schwarz, ``Robust full-dimension {MIMO} transmission based on limited
  feedback angular-domain {CSIT},'' \emph{EURASIP Journal on Wireless
  Communications and Networking}, vol. 2018, no.~1, p.~58, Dec. 2018.

\bibitem{zhang_performance_2018}
C.~Zhang, Y.~Jing, Y.~Huang, and L.~Yang, ``Performance {Analysis} for
  {Massive} {MIMO} {Downlink} {With} {Low} {Complexity} {Approximate}
  {Zero}-{Forcing} {Precoding},'' \emph{IEEE Transactions on Communications},
  vol.~66, no.~9, pp. 3848--3864, Sep. 2018.

\bibitem{kashyap_frequency-domain_2016}
S.~Kashyap, C.~Mollen, E.~Bjornson, and E.~G. Larsson, ``Frequency-domain
  interpolation of the zero-forcing matrix in massive {MIMO}-{OFDM},'' in
  \emph{Proc. IEEE {SPAWC}'16}, Edinburgh, United Kingdom, Jul. 2016, pp. 1--5.

\bibitem{li_decentralized_2017}
K.~Li, R.~R. Sharan, Y.~Chen, T.~Goldstein, J.~R. Cavallaro, and C.~Studer,
  ``Decentralized {Baseband} {Processing} for {Massive} {MU}-{MIMO}
  {Systems},'' \emph{IEEE Journal on Emerging and Selected Topics in Circuits
  and Systems}, vol.~7, no.~4, pp. 491--507, Dec. 2017.

\bibitem{alkhateeb_multi-layer_2017}
A.~Alkhateeb, G.~Leus, and R.~W. Heath, ``Multi-{Layer} {Precoding}: {A}
  {Potential} {Solution} for {Full}-{Dimensional} {Massive} {MIMO} {Systems},''
  \emph{IEEE Transactions on Wireless Communications}, vol.~16, no.~9, pp.
  5810--5824, Sep. 2017.

\bibitem{ribeiro_double-sided_2019}
L.~N. Ribeiro, S.~Schwarz, and A.~L.~F.~de Almeida, ``Double-{Sided} {Massive}
  {MIMO} {Transceivers} for {mmWave} {Communications},'' \emph{IEEE Access},
  vol.~7, pp. 157\,667--157\,679, 2019.

\bibitem{ribeiro_identification_2015}
L.~N. Ribeiro, A.~L.~F. de~Almeida, and J.~C.~M. Mota, ``Identification of
  separable systems using trilinear filtering,'' in \emph{Proc. {IEEE}
  {CAMSAP}'15}, Cancún, Mexico, Dec. 2015, pp. 189--192.

\bibitem{ribeiro_low-complexity_2017}
L.~N. Ribeiro, S.~Schwarz, M.~Rupp, A.~L.~F. de~Almeida, and J.~C.~M. Mota, ``A
  low-complexity equalizer for massive {MIMO} systems based on array
  separability,'' in \emph{Proc. {EUSIPCO}'17}, Kos, Greece, Aug. 2017, pp.
  2453--2457.

\bibitem{ribeiro_separable_2019}
L.~N. Ribeiro, A.~L.~F.~de Almeida, and J.~C. M.~Mota, ``Separable linearly
  constrained minimum variance beamformers,'' \emph{Signal Processing}, vol.
  158, pp. 15--25, May 2019.

\bibitem{ribeiro_low-rank_2019}
------, ``Low-{Rank} {Tensor} {MMSE} {Equalization},'' in \emph{Proc.
  {ISWCS}'19}, Oulu, Finland, Aug. 2019, pp. 511--516.

\bibitem{wang_two-dimensional_2017}
Z.~Wang, W.~Liu, C.~Qian, S.~Chen, and L.~Hanzo, ``Two-{Dimensional}
  {Precoding} for 3-{D} {Massive} {MIMO},'' \emph{IEEE Transactions on
  Vehicular Technology}, vol.~66, no.~6, pp. 5485--5490, Jun. 2017.

\bibitem{zhu_hybrid_2017}
G.~Zhu, K.~Huang, V.~K.~N. Lau, B.~Xia, X.~Li, and S.~Zhang, ``Hybrid
  {Beamforming} via the {Kronecker} {Decomposition} for the {Millimeter}-{Wave}
  {Massive} {MIMO} {Systems},'' \emph{IEEE Journal on Selected Areas in
  Communications}, vol.~35, no.~9, pp. 2097--2114, Sep. 2017.

\bibitem{van_trees_optimum_2002}
H.~L. Van~Trees, \emph{Optimum array processing: {Part} {IV} of detection,
  estimation and modulation theory}.\hskip 1em plus 0.5em minus 0.4em\relax
  Wiley Online Library, 2002, vol.~1.

\bibitem{truong_effects_2013}
K.~T. Truong and R.~W. Heath, ``Effects of channel aging in massive {MIMO}
  systems,'' \emph{Journal of Communications and Networks}, vol.~15, no.~4, pp.
  338--351, Aug. 2013.

\bibitem{bjornson_optimal_2014}
E.~Bjornson, M.~Bengtsson, and B.~Ottersten, ``Optimal {Multiuser} {Transmit}
  {Beamforming}: {A} {Difficult} {Problem} with a {Simple} {Solution}
  {Structure} [{Lecture} {Notes}],'' \emph{IEEE Signal Processing Magazine},
  vol.~31, no.~4, pp. 142--148, Jul. 2014.

\bibitem{bdhaardt}
Q.~H. {Spencer}, A.~L. {Swindlehurst}, and M.~{Haardt}, ``Zero-forcing methods
  for downlink spatial multiplexing in multiuser {MIMO} channels,'' \emph{IEEE
  Transactions on Signal Processing}, vol.~52, no.~2, pp. 461--471, 2004.

\bibitem{sidiropoulos_tensor_2017}
N.~D. Sidiropoulos, L.~De~Lathauwer, X.~Fu, K.~Huang, E.~E. Papalexakis, and
  C.~Faloutsos, ``Tensor {Decomposition} for {Signal} {Processing} and
  {Machine} {Learning},'' \emph{IEEE Transactions on Signal Processing},
  vol.~65, no.~13, pp. 3551--3582, Jul. 2017.

\bibitem{schwarz_performance_2016}
S.~Schwarz and M.~Rupp, ``Performance evaluation of low complexity double-sided
  massive {MIMO} transceivers,'' in \emph{Proc. {IEEE} {CCNC}'16}, Las Vegas,
  NV, USA, Jan. 2016, pp. 582--588.

\end{thebibliography}
